\documentclass[english,11pt,a4paper]{article}
\usepackage[english]{babel}
\usepackage{amsmath}
\usepackage{latexsym,amssymb,amsfonts,amscd,epsfig,amsthm,color,mathrsfs}

\usepackage{aliascnt}
\usepackage{graphicx}
\usepackage{tikz}
\usetikzlibrary{decorations.pathreplacing, shapes.geometric}
\usepackage{float}
\usepackage[left=2.3cm,right=2.3cm,top=2.3cm,bottom=2.3cm]{geometry}
\usepackage{ifthen}
\usepackage{authblk}
\usepackage{colonequals}
\usepackage{array}
\usepackage{todonotes}

\makeatletter
\newtheorem*{rep@theorem}{\rep@title}
\newcommand{\newreptheorem}[2]{%
\newenvironment{rep#1}[1]{%
 \def\rep@title{#2 \ref*{##1}}%
 \begin{rep@theorem}}%
 {\end{rep@theorem}}}
\makeatother

\usepackage{hyperref}
\hypersetup{
    bookmarksnumbered=true, 
    unicode=false, 
    pdfstartview={FitH}, 
    pdftitle={}, 
    pdfauthor={}, 
    pdfsubject={}, 
    pdfcreator={}, 
    pdfproducer={}, 
    pdfkeywords={}, 
    pdfnewwindow=true, 
    colorlinks=true, 
    linkcolor=blue, 
    citecolor=blue, 
    filecolor=blue, 
    urlcolor=blue 
}

\usepackage[nameinlink]{cleveref}
\usepackage{quantikz}
\usepackage{comment}
\Crefname{section}{Section}{Sections}

\newtheorem{theorem}{Theorem}[section]

\newaliascnt{definition}{theorem}
\newtheorem{definition}[definition]{Definition}
\aliascntresetthe{definition}
\Crefname{definition}{Definition}{Definitions}

\newaliascnt{lemma}{theorem}
\newtheorem{lemma}[lemma]{Lemma}
\aliascntresetthe{lemma}
\Crefname{lemma}{Lemma}{Lemmas}

\newaliascnt{proposition}{theorem}

\aliascntresetthe{proposition}
\Crefname{proposition}{Proposition}{Propositions}

\newaliascnt{corollary}{theorem}

\aliascntresetthe{corollary}
\Crefname{corollary}{Corollary}{Corollaries}

\newaliascnt{claim}{theorem}

\aliascntresetthe{claim}
\Crefname{claim}{Claim}{Claims}

\newaliascnt{example}{theorem}

\aliascntresetthe{example}
\Crefname{example}{Example}{Examples}

\newaliascnt{conjecture}{theorem}

\aliascntresetthe{conjecture}
\Crefname{conjecture}{Conjecture}{Conjectures}

\newaliascnt{aside}{theorem}

\aliascntresetthe{aside}
\Crefname{aside}{Aside}{Asides}

\newaliascnt{remark}{theorem}

\aliascntresetthe{remark}
\Crefname{remark}{Remark}{Remarks}

\def\ket#1{{\lvert}#1\rangle}

\def\bra#1{{\langle}#1\rvert}
\def\braket#1#2{{{\langle}#1\vert}#2\rangle}
\def\abs#1{\left| #1 \right|}

\def\norm#1{\left\| #1 \right\|}

\DeclareMathOperator{\E}{E}

\newcommand{\eps}{\varepsilon}

\def\O{\mathrm{O}}
\def\tO{\widetilde{\mathrm{O}}}

\newcommand{\ketbra}[2]{|{#1}\rangle\!\langle{#2}|}

\newcommand{\cA}{\mathcal{A}}
\newcommand{\cB}{\mathcal{B}}
\newcommand{\cC}{\mathcal{C}}

\newcommand{\cZ}{\mathcal{Z}}

\renewcommand{\a}{\alpha}
\renewcommand{\b}{\beta}

\newcommand{\w}{\omega}

\newcommand{\spann}[1]{\mathrm{span}\left\{ #1 \right\}}
\newcommand{\spaan}[1]{\mathrm{span}\left( #1 \right)}

\newcommand{\set}[1]{\left\{ #1 \right\}}
\newcommand{\inparen}[1]{\left( #1 \right)}
\newcommand{\insquare}[1]{\left[ #1 \right]}

\newcommand{\blackcircle}{\text{\textbullet}}

\newcommand{\smallblacksquare}{\text{\scalebox{0.6}{\ensuremath{\blacksquare}}}}
\newcommand{\blackdiamond}{\text{\scalebox{0.7}{\ensuremath{\blacklozenge}}}}
\newcommand{\scalestar}{\text{\scalebox{1.1}{\ensuremath{\star}}}}

\title{Loop Composition in Quantum Algorithms}
\author[1,2]{Stacey Jeffery}
\author[2]{Manideep Mamindlapally}
\author[]{Alex Baudoin Nguetsa Tankeu}
\affil[1]{QuSoft \& CWI, Amsterdam}
\affil[2]{University of Amsterdam}
\date{}

\begin{document}

\maketitle

\begin{abstract}
The quantum circuit model essentially treats every quantum algorithm as a \emph{straight-line program}. While this view is universal, recent work has shown that it is inconvenient for using different-length quantum subroutines in superposition. Using the quantum walk formalism of quantum algorithms, it is possible to model such \emph{branching} behaviour, and get better composition in this setting. 

We apply the above branching composition to Grover's algorithm, which gives a variable-time quantum search algorithm that is \emph{worse} than previous work. The reason it is worse is because branching composition does not take into account another deviation from straight-line programs: looping. We show that by modifying branching composition to also include looping, we can get a complexity that matches previous work. This highlights the importance of properly modeling the program control flow when designing quantum algorithms. 
\end{abstract}


\section{Introduction}

Any deterministic classical algorithm can be modeled as a \emph{straight-line program}, which is a sequence of operations, applied one after the other, starting from some initial state, and ending in some final state, which should encode the program's output. From this description, it is also clear that a (unitary) quantum circuit, which is the most commonly used model of quantum algorithms, is also a straight-line program, but one where we allow quantum operations, on quantum states. 

\begin{figure}[ht]
\centering
\begin{tikzpicture}[
    spA/.style={circle, draw, thick, minimum size=6mm, inner sep=1pt
    },
    spB/.style={circle, minimum size=6mm, inner sep=1pt
    },
    elabel/.style={
    fill=white, inner sep=1pt},
    slabel/.style={
    },
    >={Stealth[length=5pt]}
]

\node[spA] (P1) at (0, 0) {};
\node[spA] (P2) at (1.5, 0) {};
\node[spA] (P3) at (3, 0) {};
\node[spB] (Pdots) at (4.5, 0) {$\dots$};
\node[spA] (Plast) at (6, 0) {};

\draw[thick, ->] (-1.5,0) -- (P1) node[midway, above] {\small input};
\draw[thick, ->] (P1) -- (P2);
\draw[thick, ->] (P2) -- (P3);
\draw[thick, ->] (P3) -- (Pdots);
\draw[thick, ->] (Pdots) -- (Plast);
\draw[thick, ->] (Plast) -- (7.5,0) node[midway, above] {\small output};

\end{tikzpicture}
\caption{Visualization of a straight-line program.}\label{fig:straight-line-program}
\end{figure}
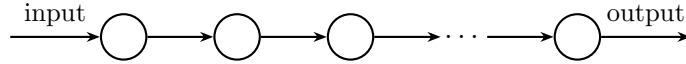

While the quantum circuit model is universal, one drawback is that it hides basically all structure of the program, aside from its length. One example of where this can give rise to a concrete disadvantage is in subroutine composition. If one of the operations represented by a node in the graph in \Cref{fig:straight-line-program} is something complicated, to be implemented by a subroutine, then this subroutine can also be represented by a straight-line program, which can then be inserted to replace the edge. However, imagine, as is often the case in a quantum algorithm, that this operation is applied to a superposition, and its behaviour is controlled on which branch of the superposition it is running in. One way to think of such a subroutine is as $N$ quantum circuits, specified by sequences of unitaries (i.e. straight-line programs): $\{U_{T_i}^i\dots U_1^i\}_{i=1}^N$, 
where we want to apply the $i$-th circuit in the $i$-th superposition branch. 
The superposition structure is completely lost in the straight-line program picture. To use this subroutine, we need to express it as a single straight-line program, that calls $U_{T_i}^i\dots U_1^i$ controlled on some particular register containing value $\ket{i}$. Even in the quantum random access model, this represents a straight-line program of length $\max_i T_i$. 

To get around this, we need to enhance our straight-line programs to allow them to express superposition branching behaviour, which is quite easy to do visually -- see \Cref{fig:branching} -- and somewhat more technical to do mathematically.

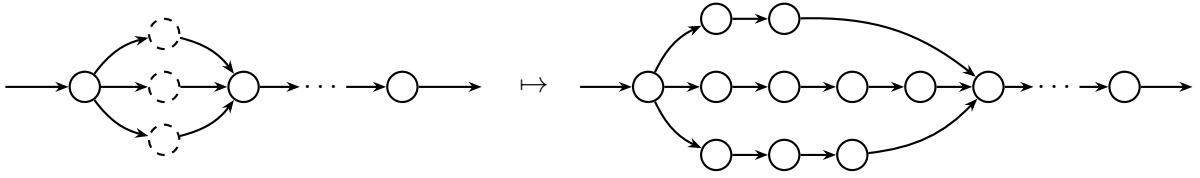
\begin{figure}[ht]
\centering
\begin{tikzpicture}
\node at (1.5,0){
\begin{tikzpicture}[scale=.7,
    spA/.style={circle, draw, thick, minimum size=4mm, inner sep=1pt
    },
    spAd/.style={circle, draw, thick, dashed, minimum size=4mm, inner sep=1pt
    },
    spB/.style={circle, minimum size=6mm, inner sep=1pt
    },
    elabel/.style={
    fill=white, inner sep=1pt},
    slabel/.style={
    },
    >={Stealth[length=5pt]}
]
\node[spA] (P1) at (0, 0) {};
\node[spAd] (P2a) at (1.5, 1) {};
\node[spAd] (P2b) at (1.5, 0) {};
\node[spAd] (P2c) at (1.5, -1) {};
\node[spA] (P3) at (3, 0) {};
\node[spB] (Pdots) at (4.5, 0) {$\dots$};
\node[spA] (Plast) at (6, 0) {};
\draw[thick, ->] (-1.5,0) -- (P1);
\draw[thick, ->] (P1) to[bend left=20] (P2a);
\draw[thick, ->] (P1) -- (P2b);
\draw[thick, ->] (P1) to[bend right=20] (P2c);
\draw[thick, ->] (P2a) to[bend left=20] (P3);
\draw[thick, ->] (P2b) -- (P3);
\draw[thick, ->] (P2c) to[bend right=20] (P3);
\draw[thick, ->] (P3) -- (Pdots);
\draw[thick, ->] (Pdots) -- (Plast);
\draw[thick, ->] (Plast) -- (7.5,0);
\end{tikzpicture}};

\node at (5.34,0) {$\mapsto$};

\node at (10,0){
\begin{tikzpicture}[scale=.6,
    spA/.style={circle, draw, thick, minimum size=4mm, inner sep=1pt
    },
    spB/.style={circle, minimum size=4mm, inner sep=1pt
    },
    elabel/.style={
    fill=white, inner sep=1pt},
    slabel/.style={
    },
    >={Stealth[length=5pt]}
]
\node[spA] (P1) at (0, 0) {};

\node[spA] (T1) at (1.5, 1.5) {};
\node[spA] (T2) at (3.0, 1.5) {};

\node[spA] (M1) at (1.5, 0) {};
\node[spA] (M2) at (3.0, 0) {};
\node[spA] (M3) at (4.5, 0) {};
\node[spA] (M4) at (6.0, 0) {};

\node[spA] (B1) at (1.5, -1.5) {};
\node[spA] (B2) at (3.0, -1.5) {};
\node[spA] (B3) at (4.5, -1.5) {};

\node[spA] (P3) at (7.5, 0) {};
\node[spB] (Pdots) at (9, 0) {$\dots$};
\node[spA] (Plast) at (10.5, 0) {};

\draw[thick, ->] (-1.5,0) -- (P1); 

\draw[thick, ->] (P1) to[bend left=20] (T1);
\draw[thick, ->] (P1) -- (M1);
\draw[thick, ->] (P1) to[bend right=20] (B1);

\draw[thick, ->] (T1) -- (T2);

\draw[thick, ->] (M1) -- (M2);
\draw[thick, ->] (M2) -- (M3);
\draw[thick, ->] (M3) -- (M4);

\draw[thick, ->] (B1) -- (B2);
\draw[thick, ->] (B2) -- (B3);

\draw[thick, ->] (T2) to[bend left=20] (P3);
\draw[thick, ->] (M4) -- (P3);
\draw[thick, ->] (B3) to[bend right=20] (P3);

\draw[thick, ->] (P3) -- (Pdots);
\draw[thick, ->] (Pdots) -- (Plast);
\draw[thick, ->] (Plast) -- (12,0); 
\end{tikzpicture}};
\end{tikzpicture}
\caption{Visualization of superposition branching, and a subroutine whose behaviour differs per branch of the superposition.}\label{fig:branching}
\end{figure}

This was first done in~\cite{jeffery2023subroutines}, where the following was proven (see \Cref{sec:straight-line-comp} for a formal statement):
\begin{theorem}[Informal]\label{thm:comp-informal}
Fix any quantum algorithm ${\cal A}$, defined by $U_1,\dots,U_{\sf L}$, that computes $g:[M]\rightarrow\{0,1\}$, where each $U_t$ represents either a basic operation, or a call to a subroutine computing a function $f:[N]\rightarrow\{0,1\}$. Fix a quantum algorithm that computes $f$, using $T_i$ basic operations, on input $i\in [N]$. Then there is a quantum algorithm that computes $g$ with bounded error in complexity
$${\sf Q}\cdot {\sf T}_{\mathrm{avg}}+{\sf L}$$
(neglecting log factors), where ${\sf Q}$ is the number of subroutine calls made by $\cal A$, and ${\sf T}_{\mathrm{avg}}=\sum_i \bar{q}_i T_i$ is a weighted average of the complexities $T_i$. (See \Cref{thm:composition-formal} for the definition of $\bar{q}_i$, and other details). 
\end{theorem}

This is the quantum analogue of what one takes for granted in classical composition of randomized programs: when a subroutine's running time depends on some random variable, the complexity of a subroutine call scales with the average (expected) running time of the subroutine, not the maximum. 

\paragraph{Variable-time quantum search.} In fact, the use of quantum subroutines where the precise running time depends on which branch of the superposition you are in had previously been studied in the context of quantum search. In standard \emph{black-box search}, one is given an oracle $f:[N]\rightarrow\{0,1\}$, and wants to decide if there exists $i\in [N]$ such that $f(i)=1$. The setting of \emph{variable-time} (quantum) search models the case where this oracle is instantiated by a subroutine for $f$ that takes time $T_i$ on input $i$. Specifically, it is assumed that the subroutine is given via an operator that can apply any $U_t^i$ controlled on $\ket{t}\ket{i}$ in some auxiliary registers. We describe this more precisely in \Cref{sec:VT-model}, but it is similar to the quantum random access model (see also \cite[Section~2.2]{jeffery2022kDist}). For simplicity, throughout this introduction, we will consider a promise version of variable-time search, where in the case that there exists $m\in [N]$ such that $f(m)=1$, it is unique. 

Naively, if we simply take Grover's algorithm and instantiate the oracle with the given subroutine, the complexity of variable-time search is: 
\begin{equation}\label{eq:naive-VT}
    \textrm{Naive variable-time search}: \tO\left(\sqrt{N}\max_{i\in [N]}T_{i}\right).
\end{equation}
However, in \cite{ambainis2010VTSearch}, Ambainis showed how to accomplish this task in complexity:
\begin{equation}\label{eq:amb-VT}
    \textrm{$\ell_2$-variable-time search}: \widetilde{O}\left(\sqrt{\sum_{i\in [N]}T_i^2}\right).
\end{equation}
Throughout this work, we let $\widetilde{O}(g(N))=g(N)\cdot \text{poly}(\log g(N))$. 

The expression in~\eqref{eq:amb-VT} is better than the naive expression in \eqref{eq:naive-VT}, 
since $\sqrt{\sum_iT_i^2}=\sqrt{N \frac{1}{N}\sum_iT_i^2}$, and $\sqrt{\frac{1}{N}\sum_iT_i^2}$ is the $\ell_2$-average of the running times, which is always at most the maximum, and can be arbitrarily smaller. An alternative algorithm achieving \eqref{eq:amb-VT}, with a matching lower bound, can be found in~\cite{ambainis2023ImprovedAlgorithmLower}.

\paragraph{Our contribution.} The results of \Cref{thm:comp-informal} apply to a much wider class of quantum algorithms than the variable-time search result, but is the variable-time search upper bound in \eqref{eq:amb-VT} a special case of this theorem, where the outer algorithm is simply Grover's optimal quantum algorithm for search~\cite{grover1996QSearch}? In \Cref{sec:grover}, we prove that it is not, by analyzing the complexity of variable-time search if we address it using \Cref{thm:comp-informal} on Grover's algorithm. Our analysis, which is tight up to constants, gives an upper bound no better than the naive upper bound in \eqref{eq:naive-VT}. 

This leads us to ask the question: Why does \Cref{thm:comp-informal} get a suboptimal result in the context of search? Is this because it's suboptimal in general? We theorize that what is actually suboptimal is \emph{treating Grover's algorithm as a straight-line program}. Grover's algorithm essentially loops through the same two reflections some sufficiently large number $O(\sqrt{N})$ of times, so it is not very natural to view it as a line of this length, and in fact, unrolling it into such a line appears to make composition work worse, by not allowing interference to occur between different iterations of the loop. 

To support this speculation, in \Cref{sec:loop-composition}, we show how to modify the program composition of \cite{jeffery2023subroutines} described in \Cref{thm:comp-informal} -- which we call  \emph{straight-line composition} -- to what we call \emph{loop composition}, where we compose subroutines with an outer algorithm that is viewed, not as a straight line (with branching), but a kind of loop (also incorporating branching). In this way, we recover the variable-time search upper bound of~\eqref{eq:amb-VT}, as well as several other known upper bounds~\cite{jeffery2023subroutines}. 

Our variable-time search upper bounds are not new, and not a contribution in and of themselves. We feel the contribution of loop composition is instead a conceptual one. \cite{jeffery2023subroutines} showed that if one always views quantum algorithms as straight-line programs, one fails to adequately capture their superposition branching behaviour, and shows how to capture this in the context of program composition. We add to this picture by showing that it is also sub-optimal to view all quantum programs as straight lines with superposition branching: one also sometimes needs to capture looping behaviour, and we show how to capture this within the same model as~\cite{jeffery2023subroutines}.

\subsection{Comparison with other work}

The composition results of \Cref{thm:comp-informal} have been improved to apply not only to quantum algorithms computing a single bit, using a new model called \emph{transducers}~\cite{belovs2024Taming}. The quantum-walk based algorithms discussed in \cite{jeffery2023subroutines} and here are a special case of transducers, but they are still very relevant in that they give some structure to the otherwise quite abstract model of transducers (see also~\cite{jeffery2024subspace} for a way of formalizing this use of graph structure). The moral of this work therefore applies to transducers as well. 

\cite{jeffery2023subroutines} also studies variable-time search, but through the lens of quantum walks. Our results in~\Cref{sec:loop-composition} are structurally similar to these, but are more direct in that they do not require the setting of a quantum walk on a graph, and are instead directly comparable to the straight-line program results mentioned in~\Cref{thm:comp-informal}.

\section{Preliminaries}

We first mention some notation. For a positive integer $N$, we let $[N]\colonequals\{1,\dots,N\}$, and $[N]' \colonequals \set{0} \cup [N]$. 
We use the notation $\ket{\varphi}$ for a vector, which is not necessarily normalized. If $\ket{\varphi_1}, \ket{\varphi_2}, \cdots$ are some vectors, $\spann{\ket{\varphi_1}, \ket{\varphi_2}, \cdots}$ denotes the span of the vectors $\ket{\varphi_1}, \ket{\varphi_2}, \cdots$.
If $\Phi$ is a set of vectors, $\spaan{\Phi}$ denotes the span of all the vectors in $\Phi$.

\subsection{Black-box search}\label{sec:black-box}

The abstract problem \emph{black-box search} models virtually any unstructured search problem. Its input is a \emph{black box} or \emph{oracle}, $f:[N]\rightarrow\{0,1\}$, by which we mean, we assume the ability to evaluate $f$ on any input $i\in [N]$. For quantum algorithms, we assume this means we can implement the unitary:
$$U_f:=2\sum_{i\in [N]:f(i)=0}\ket{i}\bra{i}-I:\ket{i}\mapsto \left\{\begin{array}{ll}
\ket{i} & \mbox{if }f(i)=0\\
-\ket{i} & \mbox{if }f(i)=1.
\end{array}\right.$$
Then our goal is, using calls to such a unitary, to decide if there exists $i\in [N]$ such that $f(i)=1$. This is the \emph{decision version} of black-box search, which we will consider in this paper. Alternatively we could consider the problem of \emph{finding} a marked element if it exists, but it is not difficult to reduce one version of the problem to the other, so we lose little by restricting to the simpler decision version.

\subsection{Variable-time quantum algorithms}\label{sec:VT-model}

Here we define variable-time algorithms following the notation of~\cite{jeffery2023subroutines}, which is based closely on the definition of~\cite{ambainis2010VTAA}.

\paragraph{Quantum algorithms.} A quantum algorithm is defined by a sequence of unitaries $U_1,\dots,U_{\sf T}$ acting on a space 
$$H_I\otimes H_A\otimes H_Z=\mathrm{span}\{\ket{i}\ket{a}\ket{z}:i\in {\cal I},a\in {\cal A},z\in {\cal Z}\}$$
for some finite sets ${\cal I}$, ${\cal A}$, and ${\cal Z}$. 

The space $H_I$ is the \emph{input} space, and we assume each unitary is controlled on this space, but does not change it, meaning that for each $t\in [{\sf T}]$, we can express
$$U_t=\sum_{i\in {\cal I}}\ket{i}\bra{i}\otimes U_t^i$$
for some unitaries $U_t^i$ acting on $H_A\otimes H_Z$. 
The unitaries $U_t^i$ are assumed to be taken from some set of ``basic'' unitaries -- for example, 2-local gates from some finite gate set, or the set of all constant depth circuits in some fixed gate set -- such that each unitary represents a single ``step'' of the algorithm.
Throughout this paper, we will always have ${\cal I}=[N]'=\{0\}\cup [N]$. 

The space $H_A$ is the \emph{answer space}, and throughout this paper, we will always have ${\cal A}=\{0,1\}$. We say the algorithm computes $f:{\cal I}\rightarrow{\cal A}$ with bounded error $\eps$ if for all $i\in {\cal I}$, the probability of measuring $f(i)$ in register $A$ after running the algorithm on input $i$ is at least $1-\eps$:
$$\norm{(I_I\otimes \ket{f(i)}\bra{f(i)}\otimes I_Z)U_{\sf T}\dots U_1\ket{i,0,0}}^2 = \norm{(\ket{f(i)}\bra{f(i)}\otimes I_Z)U_{\sf T}^i\dots U_1^i\ket{0,0}}^2\geq 1-\eps.$$
The space $H_Z$ is the workspace of the algorithm, and it can be somewhat arbitrary.

\paragraph{Model.} As is standard in the setting of variable-time algorithms, we assume that we can implement the unitary $\sum_{i,t}\ket{i}\bra{i}\otimes \ket{t}\bra{t}\otimes U_t^i$ in unit cost. This is not generally true in the gate model -- except for extremely structured cases, it is not possible to implement such a unitary in $O(1)$ local gates. However, this is reasonable in the quantum random access model, with some minimal structure on the $\{U_t^i\}_{i,t}$. A detailed discussion can be found in~\cite[Section 2.2]{jeffery2022kDist}.

\paragraph{Variable-time quantum algorithms.} Informally, a \emph{variable-time} quantum algorithm is a quantum algorithm where, at each step between $1$ and ${\sf T}$, there is some probability of the algorithm terminating. For the algorithm to have the possibility of terminating at some intermediate step $t$, it is necessary to perform a measurement at this step, to get some result that would indicate that the answer has been found, and there is no need to run the remaining computation steps. This is in contrast to the standard paradigm, where we usually run the algorithm for the full $\sf T$ steps before measuring the answer. 

Thus, formally, a variable-time quantum algorithm is a quantum algorithm that is augmented with two-outcome measurements $\{\Pi_{\leq t}\}_{t=1}^{{\sf T}}$, where if we measure $\Pi_{\leq t}$ after applying $U_t$, then we know the algorithm is ``done'', and we can measure the answer right away. In order to ensure that these measurements don't change the behaviour of the algorithm, we define them from subspaces 
\begin{equation}\label{eq:H-spaces}
\{0\}=H_0\subseteq H_1\subseteq \dots\subseteq H_{\sf T}=H_Z
\end{equation}
such that $U_t$ leaves $H_{t-1}$ invariant: 
$$U_t\Pi_{\leq t-1} = \Pi_{\leq t-1},$$ 
where $\Pi_{\leq t-1}$ is the orthogonal projector onto $H_{t-1}$. 
Since future unitaries don't change the ``done'' part of the state, the output behaviour of the algorithm is the same whether you perform the intermediate measurements or simply measure the answer register at the end -- though of course, if you perform the intermediate measurements, you might halt after significantly less than ${\sf T}$ steps.

We will assume that the measurements commute with the computational basis of $H_Z$, meaning that for each $t\in [{\sf T}]'$, there exists a set ${\cal Z}_{\leq t}\subseteq {\cal Z}$ such that 
$$H_t=\mathrm{span}\{\ket{z}:z\in {\cal Z}_{\leq t}\}.$$
Then by \eqref{eq:H-spaces}, we must have 
$$\emptyset = {\cal Z}_{\leq 0}\subseteq {\cal Z}_{\leq 1}\subseteq\dots\subseteq {\cal Z}_{\leq {\sf T}}={\cal Z}.$$
Defining ${\cal Z}_t = {\cal Z}_{\leq t}\setminus {\cal Z}_{\leq t-1}$, we get $\overline{H}_t:= H_t\cap H_{t-1}^\bot=\mathrm{span}\{\ket{z}:z\in {\cal Z}_t\}$. Finally, define ${\cal Z}_{>t}={\cal Z}\setminus {\cal Z}_{\leq t}$. 

\paragraph{Example.} For clarity, we give an example of the form a variable-time algorithm might take. Imagine a quantum algorithm consisting of multiple blocks, where after each block you measure, and if your measurement is successful, you halt, and otherwise, you move on to the next block. An example of this is a quantum search algorithm where you don't know the number of marked elements, so you first run Grover's algorithm repeatedly using about $\sqrt{2},\sqrt{4},\sqrt{8},\dots,\sqrt{N}$ iterations in successive rounds. After each round, you measure, and if the measurement result is a marked element, you halt, but otherwise, you continue with the next round. After the last round of $\sqrt{N}$ iterations, if you have still not found a marked element, you conclude there is none. 

More generally, suppose you have a quantum algorithm consisting of $B$ blocks: 
$$\underbrace{\tilde U_{N_1+\dots+N_B}\dots\tilde U_{N_1+\dots+N_{B-1}+1}}_{\mbox{Block $B$}},
\dots,\underbrace{\tilde U_{N_1+N_2}\dots \tilde U_{N_1+1}}_{\mbox{Block 2}}, \underbrace{\tilde U_{N_1}\dots \tilde U_1}_{\mbox{Block 1}}$$
where ${\sf T}=N_1+\dots+N_B$, and each $\tilde U_t$ acts on a space 
$$H_I\otimes H_A \otimes H_{Z'} = \mathrm{span}\{\ket{i}\ket{a}\ket{z}:i\in {\cal I}, a\in {\cal A}, z\in {\cal Z}'\},$$
and the idea is to run Block 1, perform some measurement $\{\Pi,I-\Pi\}$ of $H_{Z'}$, and then if the outcome is 1, measure $H_A$ and output the resulting answer, and otherwise continue with Block 2, repeat the same measurement, etc. This is not a variable-time quantum algorithm as described above, because the measurement generally changes the outcome. We can make it into a variable-time quantum algorithm by adding two extra registers, to get
$$H_Z = H_{Z'}\otimes H_F \otimes H_C = H_{Z'}\otimes\mathrm{span}\{\ket{0},\ket{1}\}\otimes\mathrm{span}\{\ket{k}:k\in \{0,\dots,B\}\}.$$
Both additional registers are initialized to $\ket{0}$. The register $F$ is a flag register to indicate that the algorithm is done, and the register $C$ is incremented at each block, until the algorithm is done, at which point it stops being incremented. That is, for each $t\in [{\sf T}]$, we will replace $\tilde U_t$ as follows. If $t$ is not equal to $N_1+\dots+N_k$ for any $k\in [B]$, meaning $\tilde U_t$ is not the last unitary in a block, then $U_t$ will simply apply $\tilde U_t$ to $IAZ'$, controlled on $\ket{0}$ in the flag register $F$.
Otherwise, if $t=N_1+\dots+N_k$, $U_t$ will consist of several sub-operations:
\begin{enumerate}
    \item Controlled on $\ket{0}$ in $F$, apply $\tilde U_t$ to $IAZ'$.
    \item Controlled on $\ket{k-1}$ in register $C$, coherently perform the measurement $\{\Pi,I-\Pi\}$, and xor the result into $F$.
    \item Controlled on $\ket{0}$ in $F$, increment the register $C$.
\end{enumerate}
An important observation is that for any $t\in \{N_1+\dots+N_{k-1},\dots,N_1+\dots+N_{k}-1\}$ for $k\in [B]$, after applying $U_t$, the state of the algorithm in the last two registers is in
$$\mathrm{span}\{\ket{0}_F\ket{k-1}_C\}\oplus \mathrm{span}\{\ket{1}_F\ket{k'}_C:k'<k-1\}.$$
That is, either the flag is 0, and the block counter register carries the number of completed blocks, $k-1$, or the flag is set to 1, meaning the algorithm is finished, and the block counter contains some $k'<k-1$, indicating the last block to terminate before setting the flag (i.e. halting). 

Define, for all $k\in [B]$, for all $t \in \{N_1+\dots+N_{k-1},\dots,N_1+\dots+N_k-1\}$,
$$H_t=\mathrm{span}\{\ket{z}_{Z'}\ket{1}_F\ket{k'}_C:z\in {\cal Z}',k'<k-1\}.$$
Finally, define:
$$H_{\sf T}=H_{N_1+\dots +N_B}=H_Z.$$
Then we have:
\begin{multline}
\{0\} = H_0=H_1=\dots=H_{N_1-1}\subset H_{N_1}=\dots = H_{N_1+N_2-1} \subset\\
\dots \subset H_{N_1+\dots+N_{B-1}}=\dots = H_{N_1+\dots+N_B-1} \subset H_{N_1+\dots+N_B}=H_Z.
\end{multline}
To see that $U_t$ leaves $H_{t-1}$ invariant, note that $\tilde U_t$ is never applied, since the flag is always set to $1$ in $H_{t-1}$. Moreover, for $t=N_1+\dots+N_k$, any state in $H_{t-1}$ is supported on values $k'<k-1$ in register $C$, so the coherent measurement that sets the flag is not performed (i.e. sub-operation 2 does nothing) and thus the flag remains set at sub-operation 3, which also then does nothing. 

Thus, this is a variable-time algorithm, and it is also easy to check that it has the same input-output behaviour as the original algorithm. 

\subsection{Straight-line subroutine composition}\label{sec:straight-line-comp}

Imagine a randomized algorithm that computes a function $g$ on input $x$ by making calls to an oracle that, on input $i$, computes $f_x(i)$.
Throughout most of this paper, we leave the input $x$ implicit, but in this section, we want to make it explicitly clear that $x$ is the input to the algorithm, and $i$ is generally an input to the oracle/subroutine. The subroutine's behaviour can depend on both of these, but because of their different statuses, we put $x$ in the subscript, and $i$ in the parentheses. 

Having fixed the oracle's behaviour, we can analyze the algorithm assuming the oracle behaves as expected, but if we ever want to actually run the algorithm, we would need to instantiate the oracle by a subroutine. Suppose we have a subroutine that implements the oracle (let us assume it does so perfectly) in $T_{i,x}$ steps on inputs $i$ and $x$, where $T_{i,x}$ is a random variable. Let $Q_{i,x}$ be the random variable representing the number of times the algorithm, with input $x$, queries the subroutine on input $i$, and assume this is independent of $T_{i,x}$. Then if $L_x$ is the number of operations made by the algorithm outside of subroutine calls, the expected running time of the algorithm on input $x$ can be computed as:
\begin{equation}\label{eq:straight-line}
\sum_i\E [Q_{i,x}]\E [T_{i,x}]+\E [L_x].
\end{equation}
Let us phrase this same statement in a slightly different way. Suppose there are ${\sf Q}$ steps at which the randomized algorithm has a nonzero probability of making a subroutine call\footnote{We assume that this is finite -- this need not be the case, even when the expected running time is finite, but we can always truncate the algorithm to run in time at most its expected running time up to a constant, introducing only bounded error. This follows from Markov's inequality, and is standard.}. For any $i$ and any $t\in\{1,\dots,{\sf Q}\}$, let $p_{i,t}(x)$ be the probability that the subroutine is called with input $i$ at the $t$-th step (only counting those steps at which there is non-zero probability of making a subroutine call). Then we can rewrite \eqref{eq:straight-line} as:
\begin{equation*}
\sum_i\sum_{t=1}^{\sf Q}p_{i,t}(x)\E [T_{i,x}]+\E [L_x] = {\sf Q}\sum_i\bar{p}_i(x)\E [T_{i,x}]+\E [L],
\end{equation*}
where $\bar{p}_i(x)=\frac{1}{\sf Q}\sum_{t=1}^{\sf Q}p_{i,t}(x)$ is the average probability of querying the subroutine on input $i$ (averaged over all steps at which subroutines are called with non-zero probability).

Something similar happens in a \emph{quantum} algorithm that makes calls to an oracle 
$$U_f=\sum_{i\in [N]}\ket{i}\bra{i}_I\otimes U_f^{(i)}+\ket{0}\bra{0}_I\otimes I$$ 
implemented by a variable-time subroutine, as long as the algorithm and subroutines are combined in a careful way~\cite{jeffery2023subroutines}. Consider a quantum algorithm that makes ${\sf Q}$ queries to $U_f$, and ${\sf L}$ additional basic operations. Let $\ket{\psi_t(x)}$ be the state of the algorithm just before the $t$-th oracle query (this will generally depend on the input $x$ to the algorithm) and let $\Pi_i$ be the projector onto $\ket{i}$ in the input register $I$. Define:
$$q_{i,t}(x)=\norm{\Pi_i\ket{\psi_t(x)}}^2
\mbox{ and }\bar{q}_i(x)=\frac{1}{\sf Q}\sum_{t=1}^{\sf Q}q_{i,t}(x).$$
We call $\bar{q}_i(x)$ the \emph{average query weight} for $i$. Then we have the following.
\begin{theorem}[\cite{jeffery2023subroutines}]\label{thm:composition-formal}
Fix a quantum algorithm that computes $g:D\rightarrow\{0,1\}$ with bounded error using calls to an oracle $U_f=\sum_i\ket{i}\bra{i}\otimes U_f^{(i)}+\ket{0}\bra{0}\otimes I$; and a variable-time quantum subroutine instantiating the oracle, as described above. 
Let ${\sf T}_{\mathrm{avg}}$ be a known value such that
$$\sum_i\bar{q}_i(x)\E [T_{i,x}]\leq {\sf T}_{\mathrm{avg}}$$
for all $x\in D$. Then there is a quantum algorithm that decides $g$ with bounded error in complexity
$$\tO\left({\sf Q}\cdot {\sf T}_{\mathrm{avg}}+{\sf L}\right).$$
\end{theorem}

This result was improved in~\cite{belovs2024Taming} to apply to a broader class of oracles, although~\cite{belovs2024Taming} does not explicitly consider algorithms whose running time is a random variable. In any case, the above theorem is sufficient for the settings of this paper. In this paper, we assume subroutines have no error. Such results can be extended to bounded-error subroutines at the cost of a log-factor overhead, using amplitude amplification, but using the techniques of~\cite{belovs2024Taming,belovs2025purification}, these log factors can be removed. 

\section{Grover's Algorithm and Subroutine Composition}\label{sec:grover}

Grover's algorithm~\cite{grover1996QSearch}, one of the most fundamental of all quantum algorithms, can be used to quadratically speed up a classical (randomized) brute force search approach to black-box search, leading to a worst-case bounded-error\footnote{In fact, Grover's algorithm even achieves \emph{one-sided error}.} quantum query complexity of $O(\sqrt{N})$, with the time complexity (counting both queries and elementary operations) worse by only a log factor. If, instead of the oracle $U_f$ (see \Cref{sec:black-box}) we are given a variable-time subroutine implementing $U_f$, as described in \Cref{sec:VT-model}, then it is not difficult to see that by running Grover's algorithm on this instantiation of the oracle, we can solve the search problem in complexity $\tO(\sqrt{N}\max_{i\in [N]}T_i)$. In this section, we investigate the complexity of an algorithm derived from instead composing Grover's algorithm with the variable-time subroutine using \Cref{thm:composition-formal}. 

For completeness, we start by describing and analysing Grover's algorithm in \Cref{sec:grover-alg}. Then we apply \Cref{thm:composition-formal} to Grover's algorithm in \Cref{sec:grover-alg-comp}. For simplicity, we restrict our attention to the setting of a unique marked vertex, as this is already enough to illustrate our point.

\subsection{Grover's algorithm}\label{sec:grover-alg}

Grover's search algorithm aims to decide the existence of a marked item among $[N]$, where the property of being ``marked'' is determined by an oracle $f:[N]\rightarrow\{0,1\}$. 

The initial state is the uniform superposition over $[N]$:
\begin{equation*}
    \ket{\pi} \colonequals \frac{1}{\sqrt{N}}\sum_{i\in [N]}\ket{i},
\end{equation*}
which can be generated by applying a Hadamard gate to each of $n=\log N$ fresh qubits $\ket{0}$. 
The algorithm then consists of alternating two reflections: 
the oracle $U_f$, which is the input to the black-box search problem; and the diffusion operator $U_{\pi}$. The first is defined in \Cref{sec:black-box} as:
\begin{align*}
    U_f &= 2\sum_{i\in [N]: f(i) = 0}\ket{i}\bra{i} - I 
\end{align*}
and we define the second:
\begin{align*}
    U_{\pi} &= 2\ket{\pi}\bra{\pi} - I = H^{\otimes n}\big(2\ket{0^n}\bra{0^n}-I\big)H^{\otimes n}.
\end{align*}
We recall that 
\begin{equation}\label{eq:oracle}
    U_f\ket{i} = \begin{cases}
        \phantom{-}\ket{i}, \hspace{0.2cm} \text{if   } f(i) = 0 \\
         -\ket{i} ,\hspace{0.2cm}  \text{if   } f(i) = 1
    \end{cases}
\end{equation}
and the action of $U_\pi$ can be deduced:
\begin{align*}
    U_\pi \ket{\pi} &= \ket{\pi} \\
    U_\pi \ket{\psi} &= -\ket{\psi} \text{ if } \braket{\pi}{\psi} = 0.
\end{align*}
The algorithm will apply $U_\pi U_f$ some number $\sf T$ of times, to get the state $(U_\pi U_f)^{\sf T}\ket{\pi}$, and then measure to get some $i\in [N]$. If $f(i)=1$, it will output $1$ (there exists a marked element) and otherwise it will output $0$. If $\sf T$ is chosen appropriately, this will give a one-sided error algorithm for black-box search.

\paragraph{Case 0: No marked item.} We first make a somewhat simple observation. If there is no marked element, then $U_f=I$, and so every application of $U_\pi U_f$ acts on $\ket{\pi}$ as the identity. Thus, for any choice of $\sf T$, $(U_\pi U_f)^{\sf T}\ket{\pi}=\ket{\pi}$, and so upon measuring this state, we will get a uniform random $i\in [N]$, and of course, we will have $f(i)=0$, since no $i$ is marked, and so we will output 0. 

\paragraph{Case 1: Marked item.} Suppose there is a unique $m\in [N]$ such that $f(m)=1$. We will argue that applying the rotation $U_\pi U_f$ moves the state towards a good state defined as: \begin{equation} \label{eq:good_state}
    \ket{\mathrm{good}} = \ket{m}.
\end{equation}
We can define the state corresponding to the superposition of all the unmarked items as: \begin{equation}\label{eq:bad_state}
    \ket{\mathrm{bad}} = \sum_{i\neq m} \frac{1}{\sqrt{N-1}}\ket{i}.
\end{equation}
Let $a$ be the unique angle in $(0,\pi/2)$ such that $\sin a=\frac{1}{\sqrt{N}}$. Then $\cos a = \sqrt{\frac{N-1}{N}}$, and so
\begin{align}
    \ket{\pi} &= \sin{a}\ket{\mathrm{good}} + \cos{a}\ket{\mathrm{bad}}.\label{eq:pi-a}
\end{align}
That is, $a$ is the angle between $\ket{\mathrm{bad}}$ and the initial state $\ket{\pi}$. Then we have the following.
\begin{lemma}\label{lem:phi-a}
    For any angle $\theta \in \big[0, \frac{\pi}{2}\big]$, define a corresponding state $\ket{\varphi_\theta} = \sin{\theta}\ket{\mathrm{good}} + \cos{\theta}\ket{\mathrm{bad}}$. Then $\ket{\varphi_a} = \ket{\pi}$, and for any $\theta\in [0,\pi/2]$, $U_\pi U_f \ket{\varphi_\theta} = \ket{\varphi_{\theta+2a}}$.
\end{lemma}
\begin{proof}
The first claim, $\ket{\varphi_a}=\ket{\pi}$, is clear from \eqref{eq:pi-a}.
For the remainder of the claim, we will prove that $U_f \ket{\varphi_\theta} = \ket{\varphi_{-\theta}}$ (Claim 1) and $U_\pi \ket{\varphi_\theta} = \ket{\varphi_{2a-\theta}}$ (Claim 2), from which the result follows.
\begin{description}
    \item[Claim 1:] For any $i\in [N]$, $U_f\ket{i} = (-1)^{f(i)}\ket{i}$. We compute:
        \begin{align*}
        U_f \ket{\varphi_\theta} &= \sin{\theta}\hspace{0.2cm} U_f\ket{\mathrm{good}} + \cos{\theta }\hspace{0.2cm}U_f\ket{\mathrm{bad}}\\
        &=\sin{\theta}\hspace{0.2cm} U_f\ket{m} + \frac{1}{\sqrt{N-1}}\cos{\theta }\hspace{0.2cm}\sum_{i \neq m} U_f\ket{i}\\
        & = -\sin{\theta}\hspace{0.2cm}\ket{m} + \frac{1}{\sqrt{N-1}}\cos{\theta }\hspace{0.2cm}\sum_{i \neq m}\ket{i}\\
        &= \sin{(-\theta)}\hspace{0.2cm}\ket{\mathrm{good}}+\cos{(-\theta) }\ket{\mathrm{bad}} = \ket{\varphi_{-\theta}},
    \end{align*}
    where we used $-\sin(\theta)=\sin(-\theta)$ and $\cos(\theta)=\cos(-\theta)$ for all $\theta$. 
    \item[Claim 2:] For the second claim, we recall that $U_\pi=2\ket{\pi}\bra{\pi}-I$, from which we can compute:
    \begin{align*}
        U_\pi \ket{\varphi_\theta} &= (2\ket{\pi}\bra{\pi} - I)(\sin{\theta}\ket{\mathrm{good}} + \cos{\theta}\ket{\mathrm{bad}})\\
        & = \sin{\theta} (2\ket{\pi}\bra{\pi} - I)\ket{\mathrm{good}} + \cos{\theta}(2\ket{\pi}\bra{\pi} - I)\ket{\mathrm{bad}}\\
        &= \sin{\theta} (2\ket{\pi}\braket{\pi}{\mathrm{good}} - \ket{\mathrm{good}}) + \cos{\theta}(2\ket{\pi}\braket{\pi}{\mathrm{bad}} - \ket{\mathrm{bad}})\\
        & = \sin{\theta} (2\sin{a}\ket{\pi} - \ket{\mathrm{good}}) + \cos{\theta} (2\cos{a}\ket{\pi}-\ket{\mathrm{bad}})\\
        & = \begin{multlined}[t][0.7\displaywidth]
            2\sin{\theta}\sin{a} (\sin{a}\ket{\mathrm{good}}+\cos{a}\ket{\mathrm{bad}}) - \sin{\theta}\ket{\mathrm{good}}\\
            +2\cos{\theta}\cos{a} (\sin{a}\ket{\mathrm{good}}+\cos{a}\ket{\mathrm{bad}})-\cos{\theta}\ket{\mathrm{bad}}
        \end{multlined}\\
        &= \begin{multlined}[t][0.7\displaywidth](2\sin{\theta}\sin^2{a} + 2\cos{\theta}\cos{a}\sin{a}-\sin{\theta})\ket{\mathrm{good}}+ (2\cos{\theta}\cos^2{a}\\
        + 2\sin{\theta}\cos{a}\sin{a} -\cos{\theta})\ket{\mathrm{bad}}
        \end{multlined}\\
        &= (-\sin{\theta}\cos{2a} + \cos{\theta}\sin{2a})\ket{\mathrm{good}} + (\sin{\theta}\sin{2a} + \cos{\theta}\cos{2a})\ket{\mathrm{bad}}\\
        & = \sin{(2a-\theta)}\ket{\mathrm{good}} + \cos{(2a-\theta)}\ket{\mathrm{bad}} = \ket{\varphi_{2a-\theta}}.
    \end{align*}
\end{description}    
Combining these two calculations, we can deduce the following:
\begin{equation*}
    U_\pi U_f \ket{\varphi_\theta} = U_\pi\ket{\varphi_{-\theta}} = \ket{\varphi_{2a -(-\theta)}} = \ket{\varphi_{\theta+2a}}. \qedhere 
\end{equation*}
\end{proof}

We have all the tools to compute an appropriate value for ${\sf T}$. 
After ${\sf T}$ iterations, we have: 
\begin{align}
    (U_\pi U_f)^{\sf T} \ket{\pi}=
    (U_\pi U_f)^{\sf T} \ket{\varphi_a} &=\ket{\varphi_{a + 2{\sf T}a}} = \ket{\varphi_{(2{\sf T}+1)a}} \nonumber\\
    &= \sin{((2{\sf T}+1)a)}\ket{\mathrm{good}} + \cos{((2{\sf T}+1)a)}\ket{\mathrm{bad}}. \label{eq:final_state}
\end{align}

To minimize the distance from the $\ket{\mathrm{good}}$ state, its corresponding amplitude should be close to 1. In other words, we want $\sin{(2{\sf T}+1)a}\approx 1$ which can be achieved by ensuring $
    (2{\sf T}+1)a \approx \frac{\pi}{2}$. 
As $\sin{a}=\frac{1}{\sqrt{N}}$, for large $N$, $a \approx \frac{1}{\sqrt{N}}$. Thus, we take ${\sf T}=\lfloor\frac{\pi}{4}\sqrt{N}\rceil$. 
By the Maclaurin series expansion of $\sin a=\frac{1}{\sqrt{N}}$, we have:
\begin{align*}
    \frac{1}{\sqrt{N}}=\sin a\leq & a \leq \sin a + \frac{a^3}{3!} = {\frac{1}{\sqrt{N}}} + o\left(\frac{1}{N}\right) \\
    \frac{2{\sf T}+1}{\sqrt{N}}\leq &(2{\sf T}+1)a \leq \frac{2{\sf T}+1}{\sqrt{N}} +o\left(\frac{1}{\sqrt{N}}\right)\\
    \frac{2(\frac{\pi}{4}\sqrt{N}-\frac{1}{2})+1}{\sqrt{N}}\leq &(2{\sf T}+1)a \leq \frac{2(\frac{\pi}{4}\sqrt{N}+\frac{1}{2})+1}{\sqrt{N}}+ o\left(\frac{1}{\sqrt{N}}\right)\\
    \frac{\frac{\pi}{2}\sqrt{N}}{\sqrt{N}}\leq &(2{\sf T}+1)a \leq \frac{\frac{\pi}{2}\sqrt{N}+2}{\sqrt{N}}+ o\left(\frac{1}{\sqrt{N}}\right)\\
    \frac{\pi}{2}\leq &(2{\sf T}+1)a \leq \frac{\pi}{2}+ O\left(\frac{1}{\sqrt{N}}\right).
\end{align*}
Then we have:
\begin{align*}
    |\bra{m}(U_\pi U_f)^{\sf T}\ket{\pi}|^2 
    &=\sin^2((2{\sf T}+1)a)\\
    &\geq 1-\sin^2\left(\frac{\pi}{2}-(2{\sf T}+1)a\right)\geq 1-\left(\frac{\pi}{2}-(2{\sf T}+1)a\right)^2\geq
    1-O\left(\frac{1}{N}\right).
\end{align*}
Thus, with this choice of ${\sf T}$, we will measure $m$, and therefore output 1, with probability very close to~1. 

\subsection{Grover's algorithm with composition}\label{sec:grover-alg-comp}

We now apply \Cref{thm:composition-formal} to Grover's algorithm, and a variable-time subroutine instantiating the oracle $U_f:\ket{i} \mapsto (-1)^{f(i)}\ket{i}$. To get an upper bound on the resulting quantum algorithm for variable-time search, we need to upper bound three quantities. First, ${\sf Q}$ is the number of queries made by the outer algorithm, so in the case of Grover's algorithm, we have ${\sf Q}=O(\sqrt{N})$. Second, ${\sf L}$ is the number of basic operations outside of queries. In the case of Grover's algorithm, we need to count the $O(\sqrt{N})$ calls to the non-query reflection, $U_\pi$, each of which can be implemented in $O(\log N)$ basic operations, so ${\sf L}=O(\sqrt{N}\log N)$. Finally, and most complicated, we need to upper bound ${\sf T}_{\mathrm{avg}}$.

\begin{lemma}\label{lem:weighted-avg}
    Consider Grover's algorithm, where the input contains a single marked element $m$, as described in \Cref{sec:grover-alg}, with the oracle instantiated by a variable-time quantum subroutine. Then $\sum_{i\in [N]}\bar{q}_i \E[T_i]=\Theta\left(\sum_{i\in [N]\setminus\{m\}}\frac{1}{N-1}\E[T_i]+\E[T_m] \right)$, where $\E[T_i]$ is the expected complexity of the subroutine on input $i$, and $\bar{q}_i$ is the average query weight, defined in  \Cref{sec:straight-line-comp}. 
\end{lemma}
It follows that if we want a value ${\sf T}_{\mathrm{avg}}$ that upper bounds $\sum_i\bar{q}_i\E[T_i]$ for \emph{all} $m\in [N]$, it is necessary and sufficient to take
$${\sf T}_{\mathrm{avg}}=\Theta\left(\sum_{i\in [N]}\frac{1}{N}\E[T_i]+\max_{m\in [N]}\E[T_m]\right)=\Theta\left(\max_{i\in [N]}\E[T_i]\right).$$ 
Note that we could not get an asymptotically better bound here, since the upper bound ${\sf T}_{\mathrm{avg}}$ must hold for \emph{any} input; the marked vertex could be anything. We discuss this further in \Cref{sec:VT-bounds}.
Then by \Cref{thm:composition-formal}, we get an upper bound of 
\begin{equation}\label{eq:straight-line-2}
\widetilde{O}({\sf Q}\cdot{\sf T}_{\mathrm{avg}}+{\sf L}) = \widetilde{O}\left(\sqrt{N}\max_{i\in [N]}\E[T_i]\right)
\end{equation}
on variable-time search, which is no better than the naive upper bound of \eqref{eq:naive-VT}.

\begin{proof}[Proof of \Cref{lem:weighted-avg}]
Fix any oracle $f$ such that there is a unique marked $m\in [N]$.
As in \Cref{sec:straight-line-comp}, for any $t\in [1,\dots,\frac{\pi}{4}\sqrt{N}]$, let $\ket{\psi_{t}}$ be the state right before the $t^{\mathrm{th}}$ call to $U_f$ in Grover's algorithm (which depends on $f$). Then: 
\begin{align*}
    \ket{\psi_{t}} &= (U_\pi U_f)^{t-1}\ket{\pi} = \ket{\varphi_{a+2(t-1)a}} = \ket{\varphi_{(2t-1)a}} 
    = \sin((2t-1)a)\ket{\mathrm{good}}+\cos((2t-1)a)\ket{\mathrm{bad}}
\end{align*}
by \Cref{lem:phi-a}. Thus, for any $i\in [N]$: 
\begin{align*}
    q_{i,t}&=\norm{\Pi_i\ket{\psi_t}}^2 =
    \left\{\begin{array}{ll}
     \sin^2((2t-1)a)|\braket{i}{\mathrm{good}}|^2 & \mbox{if }f(i)=1\\
    \cos^2((2t-1)a)|\braket{i}{\mathrm{bad}}|^2 & \mbox{if }f(i)=0.
\end{array}\right.
\end{align*}
Thus, if $i=m$ is the unique marked element, then:
$$q_{m,t}=\sin^2((2t-1)a),$$
and if $i\in [N]\setminus\{m\}$,
$$q_{i,t}=\frac{\cos^2((2t-1)a)}{\sqrt{N-1}},$$
by \eqref{eq:bad_state}. Thus, we can compute the average query weight, first, for $i=m$:
\begin{align*}
    \bar{q}_m=\frac{1}{\sf Q}\sum_{t=1}^{\sf Q}q_{m,t}
    =\frac{4}{\pi\sqrt{N}}\sum_{t=1}^{\frac{\pi}{4}\sqrt{N}}\sin^2((2t-1)a)
\end{align*}
and then for $i\in [N]\setminus\{m\}$:
\begin{align*}
    \bar{q}_i=\frac{1}{\sf Q}\sum_{t=1}^{\sf Q}q_{i,t}
    =\frac{4}{\pi\sqrt{N}}\sum_{t=1}^{\frac{\pi}{4}\sqrt{N}}\frac{\cos^2((2t-1)a)}{N-1}.
\end{align*}
Then we can upper bound the average complexity of a subroutine call, as follows:
\begin{align*}
    L_f:=\sum_{i\in [N]}\bar{q}_i\E[T_i]
     &= \frac{4}{\pi\sqrt{N}}\left(\sum_{i\in [N]: i\neq m}\sum_{t=0}^{\frac{\pi}{4}\sqrt{N}-1}\frac{\cos^2{[(2t+1)a]}}{N-1} \E[T_i] + \sum_{t=0}^{\frac{\pi}{4}\sqrt{N}-1}\sin^2{[(2t+1)a]}\E[T_m]\right).
\end{align*}
Since $\forall x \in \mathbb{R}$, $\cos^2{x} = \frac{1+\cos{2x}}{2}$ and $\sin^2{x} = \frac{1-\cos{2x}}{2}$, we can continue:
\begin{align*}
    \frac{\pi}{4}\sqrt{N}L_f &= \sum_{i\in [N]: i\neq m} \frac{\E[T_i]}{N-1}\sum_{t=0}^{\frac{\pi}{4}\sqrt{N}-1}\frac{1+\cos{[(4t+2)a]}}{2} + \E[T_m]\sum_{t=0}^{\frac{\pi}{4}\sqrt{N}-1}\frac{1-\cos{[(4t+2)a]}}{2} \\
    &= \sum_{i\in [N]:i\neq m}\frac{\E[T_i]}{2(N-1)} \Bigg[\frac{\pi}{4}\sqrt{N}+\sum_{t=0}^{\frac{\pi}{4}\sqrt{N}-1}\cos{[(4t+2)a]}\Bigg]\\
    &\qquad\qquad\qquad + \frac{\E[T_m]}{2} \Bigg[\frac{\pi}{4}\sqrt{N}-\sum_{t=0}^{\frac{\pi}{4}\sqrt{N}-1}\cos{[(4t+2)a]}\Bigg]
\end{align*}
from which we have:
\begin{align}\label{eq:L-f-A}
    L_f &= 
    \sum_{i\in [N]:i\neq m}\frac{\E[T_i]}{2(N-1)} + \frac{\E[T_m]}{2}+\frac{2}{\pi\sqrt{N}} \Bigg[\sum_{i\in [N]:i\neq m}\frac{\E[T_i]}{N-1} - {\E[T_m]}\Bigg]\sum_{t=0}^{\frac{\pi}{4}\sqrt{N}-1}\cos{[(4t+2)a]}.
\end{align}
Using the identity $\cos(x+y)=\cos x \cos y - \sin x \sin y$, we can compute:
\begin{align*}
    \sum_{t=0}^{\frac{\pi}{4}\sqrt{N}-1}\cos{[(4t+2)a]}
    &=\sum_{t=0}^{\frac{\pi}{4}\sqrt{N}-1}\big(\cos{4ta}\cos{2a} - \sin{4ta}\sin{2a}\big)\\
    &=\cos{2a}\sum_{t=0}^{\frac{\pi}{4}\sqrt{N}-1}\cos{4ta} - \sin{2a}\sum_{t=0}^{\frac{\pi}{4}\sqrt{N}-1}\sin{4ta}.
\end{align*}
From here, we can use Lagrange's trigonometric identities, which state that for $\theta  \neq 0 \mod 2\pi$:
\begin{equation*}
    \sum_{k=0}^T \sin{k\theta} = \frac{\cos{\frac{1}{2}\theta}-\cos{\big((T+\frac{1}{2})\theta\big)}}
    {2\sin{\frac{1}{2}\theta}},
    \mbox{ and }
    \sum_{k=0}^T \cos{k\theta} = \frac{\sin{\frac{1}{2}\theta}+\sin{\big((T+\frac{1}{2})\theta\big)}}{2\sin{\frac{1}{2}\theta}} 
\end{equation*}
to get:
\begin{equation*}\label{eq:sum-cos}
\begin{split}
    &\sum_{t=0}^{\frac{\pi}{4}\sqrt{N}-1}\cos{[(4t+2)a]}\\
    ={}&\frac{\cos(2a)\left(\sin(2a)+\sin\left(\left({\pi}\sqrt{N}+2\right)a\right)\right) - \sin({2a})\left(\cos({2a})-\cos\left(\left({\pi}\sqrt{N}+2\right)a\right)\right)}{2\sin({2a})}\\
    ={}&\frac{\cos(2a)\sin\left(\left({\pi}\sqrt{N}+2\right)a\right)+ \sin({2a})\cos\left(\left({\pi}\sqrt{N}+2\right)a\right)}{2\sin({2a})}\\
    ={}&\frac{\sin\left((\pi\sqrt{N}+4)a\right)}{2\sin({2a})}
    =\frac{\sin\left(\pi-(\pi\sqrt{N}+4)a\right)}{2\sin({2a})}
    \leq \frac{1}{2}\frac{|\pi-\pi\sqrt{N}a-4a|}{4a/\pi}\\
    &=\frac{\pi}{8}|\pi/a-\pi\sqrt{N}-4|
\end{split}
\end{equation*}
where we used $\sin(x+y)=\cos x\sin y +\sin x\cos y$, followed by $\sin x = \sin(\pi-x)$, and $\sin x\leq |x|$ and $\sin x\geq 2x/\pi$ (whenever $x\in [0,\pi/2]$). Since
$$\frac{1}{\sqrt{N}}=\sin a \leq a \leq \frac{\pi}{2}\sin a=\frac{\pi}{2}\frac{1}{\sqrt{N}},$$
we can upper bound this by $O(\sqrt{N})$, which is sufficient for our purposes. Plugging this into \eqref{eq:L-f-A} gives:
\begin{align*}
    L_f &= 
    \sum_{i\in [N]:i\neq m}\frac{\E[T_i]}{2(N-1)} + \frac{\E[T_m]}{2}+\frac{2}{\pi\sqrt{N}} \Bigg[\sum_{i\in [N]:i\neq m}\frac{\E[T_i]}{N-1} - {\E[T_m]}\Bigg]O(\sqrt{N})\\
    &=\Theta\left(\sum_{i\in [N]\setminus\{m\}}\frac{1}{N-1}\E[T_i]+\E[T_m] \right),
\end{align*}
as claimed. 
\end{proof}

\section{Comparison of Variable-Time Search Upper Bounds}\label{sec:VT-bounds}

In this section, we consider a more general setting for search. We no longer restrict to a single marked element. Instead, we let $P_1$ be \emph{any} subset of all $f:[N]\rightarrow\{0,1\}$ such that there exists $i\in [N]$ with $f(i)=1$, and we suppose the input $f$ is promised to be the all-0s function, or a function from $P_1$. 
Moreover, we let $\pi$ be any distribution on $[N]$ from which we can easily generate a coherent quantum sample, and we let $\eps=\min_{f\in P_1}\sum_{i:f(i)=1}\pi(i)$.  
In this more general setting, the naive upper bound on variable-time search becomes:
\begin{equation}\label{eq:naive-VT-gen}
    \textrm{Naive variable-time search}: \widetilde{O}\left(\frac{1}{\sqrt{\eps}}\max_{i\in [N]}\E[T_i] \right),
\end{equation}
and Ambainis' upper bound~\cite{ambainis2010VTSearch} becomes:
\begin{equation}\label{eq:amb-VT-gen}
    \textrm{$\ell_2$-variable-time search}: \widetilde{O}\left(\frac{1}{\sqrt{\eps}}\sqrt{\sum_{i\in [N]}\pi(i)\E[T_i^2]}\right)
    =\widetilde{O}\left(\sqrt{\frac{\sum_{i\in [N]}\pi(i)\E[T_i^2]}{\min_{f\in P_1}\sum_{i:f(i)=1}\pi(i)}}\right).
\end{equation}
Ref.~\cite{jeffery2023subroutines} gave two alternative upper bounds for variable-time search that are generally incomparable to \eqref{eq:amb-VT-gen}:
\begin{align}
    \textrm{$\ell_1$-variable-time search}:\; & \widetilde{O}\left(\sqrt{\frac{\sum_{i\in [N]}\pi(i)\E[T_i]}{\min_{f\in P_1}\sum_{i:f(i)=1}\frac{\pi(i)}{T_i}}}\right)\label{eq:ell1-VT-gen}\\       \textrm{$\ell_0$-variable-time search}:\; & \widetilde{O}\left(\frac{1}{\sqrt{\min_{f\in P_1}\sum_{i:f(i)=1}\frac{\pi(i)}{\E[T_i^2]}}}\right).\label{eq:ell0-VT-gen}
\end{align}
For the sake of discussion, \Cref{tab:compare} shows the $\ell_b$ upper bounds for each $b\in\{0,1,2\}$, when we restrict to the case where $P_1$ is the set of $f:[N]\rightarrow\{0,1\}$ such that there is a unique marked element $m$ and furthermore, we are promised that $\E[T_m]\leq {\sf T}_{\max}$ for some fixed upper bound ${\sf T}_{\max}$. This is similar to the case we considered in the introduction, and \Cref{sec:grover}, except for the upper bound ${\sf T}_{\max}$. We first note that this setting is not very well motivated. If the values $\E[T_i]$ are known, we can then simply restrict our search to the subset $\{i\in [N]:\E[T_i]\leq {\sf T}_{\max}\}$. If not, then we can still simply stop running any subroutine after $c{\sf T}_{\max}$ steps, since if it were going to find a marked vertex, it would have done so by then with probability at least $1-1/c$, by Markov's inequality. However, for the purpose of discussion, this is a useful setting, since it allows us to better compare our results from \Cref{sec:grover} with other variable-time search upper bounds.

\begin{table}
\centering
\renewcommand{\arraystretch}{1.8}
\begin{tabular}{|c|c|c|}
\hline
    $\ell_2$ & $\ell_1$ & $\ell_0$ \\
    \hline
    $\sqrt{N}\sqrt{\frac{1}{N}\sum_{i\in[N]}\E[T_i^2]}$ & $\sqrt{N}\sqrt{\frac{1}{N}\sum_{i\in[N]}\E[T_i] \cdot {\sf T}_{\max}}$ & $\sqrt{N}{\sf T}_{\max}$\\
    \hline
\end{tabular}
\caption{Comparison of the $\ell_2$, $\ell_1$, and $\ell_0$-variable-time search upper bounds, in the special case where we are promised that there is a unique marked element, and its checking time is at most ${\sf T}_{\max}$.}\label{tab:compare}
\end{table}

In this setting, by \Cref{lem:weighted-avg}, our upper bound on variable-time search using straight-line composition becomes:
\begin{equation*}\label{eq:straight-line-VT}
    \textrm{Straight-line composition variable-time search}: \widetilde{O}\left(\sqrt{N}\left(\frac{1}{N}\sum_{i\in [N]}\E[T_i]+{\sf T}_{\max} \right)\right).
\end{equation*}
This is always worse than the $\ell_1$-variable-time search upper bound in \Cref{tab:compare}, since it uses arithmetic mean of the $\ell_1$-average of the running times and ${\sf T}_{\max}$, which is always at least the geometric mean, by the AM-GM inequality.

\section{Loop Composition}\label{sec:loop-composition}

The straight-line composition, \Cref{thm:composition-formal}, applied to Grover's search algorithm as done in \Cref{sec:grover-alg-comp} treats Grover's algorithm as a straight-line program, with branching in the subroutine calls. However, this ignores a key structural property of Grover's algorithm: that every step is a repetition of a single Grover iterate. As we saw in \Cref{sec:VT-bounds}, this results in a suboptimal bound.
In this section, we exploit this \emph{loop}ing structure.
Rather than composing variable-time subroutines into a straight-line unrolling of Grover's algorithm, we instead compose the subroutine into a more loop-like structure. We call this \emph{loop composition}, and it achieves the following, which matches the upper bound of Ambainis from \eqref{eq:amb-VT}.
\begin{theorem}[Informal]\label{thm:loop-composition-informal}
Fix a variable-time quantum algorithm that computes $f:[N]\rightarrow\{0,1\}$ using $T_i$ basic operations on input $i\in [N]$, for random variable $T_i$. Then there is a quantum algorithm that decides if there exists $i\in [N]$ such that $f(i)=1$ with bounded error in complexity $\tO\left(\sqrt{\sum_{i\in [N]}\E[T_i^2]}\right)$.
\end{theorem}
For a formal version of this statement, see \Cref{thm:general-picture}, where we also recover the different variable-time bounds of~\cite{jeffery2023subroutines} in \Cref{sec:VT-bounds}, and achieve something slightly better (also already known from~\cite{jeffery2023subroutines}) when the values $\E[T_i]$ are known.

In \Cref{sec:phase-estimation}, we review the algorithmic building blocks used to prove \Cref{thm:composition-formal} (straight-line composition) in \cite{jeffery2023subroutines}, as we will use the same ideas in our loop composition.
To build the intuition incrementally, we first analyze loop composition without the composition, i.e. in the simpler setting where the oracle can be queried at a single uniform cost, in \Cref{sec:simple-picture}. Finally, in \Cref{sec:general-picture} we extend the analysis to the general variable-cost setting, where the oracle is implemented by a  subroutine that has cost $T_i$ on input $i$.

\subsection{Phase estimation algorithms}
\label{sec:phase-estimation}

Loosely speaking, phase estimation~\cite{kitaev1996PhaseEst} is a quantum primitive that makes controlled calls to some unitary $U$, and, given as input some eigenstate of $U$, estimates the phase $\theta$ of the corresponding eigenvalue $e^{i\theta}$.
In this section, we review a specific type of phase estimation algorithm, where the unitary is a product of two reflections, and we are only interested in distinguishing a zero phase from non-zero phases. An algorithm of this specific structure can be built from any straight-line quantum algorithm and variable-time subroutines to prove the straight-line composition result described in \Cref{thm:composition-formal}, and we will use a similar technique to prove our loop composition result. 

We borrow the specific formulation of such a phase estimation algorithm from~\cite{jeffery2022kDist} and for the remainder of this paper, we let the phrase ``phase estimation algorithm'' refer to this specific formalism. Refer to \cite[Section~3.1]{jeffery2022kDist} for intuition about why such an algorithm works. 

A phase estimation algorithm is defined by a set of parameters, which we formalize in the following definition.
\begin{definition}[Parameters of a Phase Estimation Algorithm]\label{def:phase-est-alg}
Fix a finite-dimensional complex inner product space $H$, a unit vector $\ket{\psi_0}\in H$, and sets of pairwise orthogonal vectors
$\Psi^{\cal A},\Psi^{\cal B}\subset H$ (a pair of vectors $\ket{\psi_A}\in\Psi^{\cal A}$ and $\ket{\psi_B}\in \Psi^{\cal B}$ need not be orthogonal). We further assume that $\ket{\psi_0}$ is orthogonal to every vector in $\Psi^{\cal B}$.  Let $\Pi_{\cal A}$ be the orthogonal projector onto ${\cal A}=\spaan{\Psi^{\cal A}}$, and similarly for $\Pi_{\cal B}$.
\end{definition}
These parameters define a product of reflections:
$U_{\cal AB}=(2\Pi_{\cal A}-I)(2\Pi_{\cal B}-I)$.
The algorithm defined by $(H,\ket{\psi_0},\Psi^{\cal A},\Psi^{\cal B})$ performs phase estimation of $U_{\cal AB}$ on initial state $\ket{\psi_0}$, to sufficient precision that by measuring the phase register and checking if the output is 0, we can distinguish between a \emph{negative case} and a \emph{positive case}. For analyzing such an algorithm, it turns out to be sufficient to prove that in the positive case, there exists a \emph{positive witness}, and in the negative case, there exists a \emph{negative witness} -- we define these below.

\begin{definition}[Negative Witness]\label{def:neg-witness}
A \emph{negative witness} for $(H,\ket{\psi_0},\Psi^{\cal A},\Psi^{\cal B})$ is a pair of vectors $\ket{w_{\cal A}},\ket{w_{\cal B}}\in H$ such that 
$\ket{w_{\cA}} \in \spaan{\Psi_\cA}$, $\ket{w_\cB} \in \spaan{\Psi_\cB}$ 
and $\ket{\psi_0}=\ket{w_{\cal A}}+\ket{w_{\cal B}}$.
\end{definition}

\begin{definition}[Positive Witness]\label{def:pos-witness}
A \emph{positive witness} for $(H,\ket{\psi_0},\Psi^{\cal A},\Psi^{\cal B})$ is a vector $\ket{w}\in H$ such that $\braket{\psi_0}{w}\neq 0$ and 
$\Pi_\cA \ket{w} = \Pi_\cB \ket{w} = 0$.
\end{definition}

Then the following theorem states the existence of a general type of algorithm -- what we are calling a phase estimation algorithm -- that we will use to prove \Cref{thm:loop-composition-informal}.

\begin{theorem}[\cite{jeffery2022kDist}]\label{thm:phase-est-fwk}
Fix $(H,\ket{\psi_0},\Psi^{\cal A},\Psi^{\cal B})$ as in \Cref{def:phase-est-alg}.
Suppose we can generate the state $\ket{\psi_0}$ in cost ${\sf S}$, and implement $U_{\cal AB}=(2\Pi_{\cal A}-I)(2\Pi_{\cal B}-I)$ in cost ${\sf A}$.
Let $c_+\in [1,50]$ be some constant, and let ${\cal C}_-\geq 1$ be a positive real number (that may scale with some implicit input size). 
Suppose we are guaranteed that exactly one of the following holds:
\begin{description}
\item[Positive Condition:] There is a positive witness $\ket{w_+}$ s.t.~$\frac{|\braket{w_+}{\psi_0}|^2}{\norm{\ket{w_+}}^2}\geq \frac{1}{c_+}$.
\item[Negative Condition:] There is a negative witness $\ket{w_{\cal A}},\ket{w_{\cal B}}$ s.t.~$\norm{\ket{w_{\cal A}}}^2 \leq {\cal C}_-$.
\end{description}
Then there is a quantum algorithm that distinguishes these two cases with bounded error in cost
$$\O\left({\sf S}+\sqrt{{\cal C}_-}{\sf A}\right).$$
\end{theorem}

\subsection{The loop structure, without composition}
\label{sec:simple-picture}

In this section, we present a phase estimation algorithm for black-box search, in the sense of \Cref{def:phase-est-alg} and \Cref{thm:phase-est-fwk}. This does no better than Grover's algorithm, but its particular structure facilitates later replacement of the oracle query with a variable-time subroutine (in \Cref{sec:general-picture}). The results of this section are not required for \Cref{sec:general-picture}, but the reader may find the next section easier to follow after reading this section.

Recall from \Cref{sec:black-box} that in the problem of black-box search, the input is some $f:[N]\rightarrow\{0,1\}$, which we assume we can access via queries to a unitary $U_f$. To use such an algorithm in practice, this unitary would need to be instantiated by a subroutine, which we consider in the following section. For now, we will treat $U_f$ as a black box, that can be implemented in unit cost.

\begin{theorem}\label{thm:simple-picture}
    Fix $\mu\geq 1$. Let $f: [N] \to \set{0,1}$ be accessible through an oracle $U_f$, and let $M_f := \set{i \in [N] : f(i) = 1}$ be the set of marked elements.
    There exists a quantum phase estimation algorithm that decides whether $M_f$ is empty or $|M_f|\geq \mu$ -- i.e. that decides a promise version of black-box search -- with bounded error, in complexity
    $$\O\left(\sqrt{\frac{N}{\mu}}\log N \right).$$
\end{theorem}
In fact, it is well known how to solve black-box search by doing phase estimation of the Grover iterate -- $U_\pi U_f$ in the notation of \Cref{sec:grover} -- which is a product of two reflections. The difference here is that we construct a product of two reflections in a way that will be amenable to composition with a subroutine in \Cref{sec:loop-composition}.
\begin{proof}
To design a phase estimation algorithm, we need to fix parameters as in \Cref{def:phase-est-alg}, starting with a space
$$H=\mathrm{span}\{\ket{d,i,b}:d\in\{0,1,\circ,\bot\},i\in [N],b\in\{0,1\}\}.$$
Next we need a pair of subspaces ${\cal A}$ and ${\cal B}$ of $H$. Towards that goal, define four sets of vectors, as follows, where $\w$ is a weight that governs the norms of the positive and negative witnesses, and will be optimally tuned later. 
\begin{align*}
    \Psi_\bullet &= \left\{ \ket{\circ, 0, 0} - \sum_{i \in [N]} \sqrt{\tfrac{\w}{N}}\, \ket{0, i, 0} \right\}, \\
    \Psi_\blackdiamond &= \left\{ \ket{0,i,0} - \ket{\bot, i, 0} : i \in [N] \right\}, \\
    \Psi_\smallblacksquare &= \left\{ \ket{\bot, i, b} - \ket{1, i, b} : i \in [N],\, b \in \set{0,1} \right\}, \\
    \Psi_\scalestar &= \left\{ \ket{1, i, 0} : f(i) = 0,\, i \in [N] \right\}.
\end{align*}
Let $\Psi_\cA = \Psi_\bullet \cup \Psi_\smallblacksquare$ and $\Psi_\cB = \Psi_\blackdiamond \cup \Psi_\scalestar$. These define the spaces ${\cal A}=\spaan{\Psi_{\cal A}}$ and ${\cal B}=\spaan{\Psi_{\cal B}}$, and thus the unitary $U_{\cal AB}$ on which we do phase estimation. It remains to define the initial state, which we set to $\ket{\psi_0} = \ket{\circ, 0, 0}$.

\Cref{fig:simple-span-graph} shows the \emph{overlap graph} of the four sets of vectors we have defined. From this, it is easy to verify that $\Psi_{\cal A}$ and $\Psi_{\cal B}$ each consist of orthogonal vectors, and $\ket{\psi_0}$ is orthogonal to ${\cal B}$, as required.

\begin{figure}
\centering
\begin{tikzpicture}[
    spA/.style={circle, draw, thick, minimum size=8mm, inner sep=1pt
    },
    spB/.style={circle, draw, thick, dashed, minimum size=8mm, inner sep=1pt
    },
    elabel/.style={
    fill=white, inner sep=1pt},
    slabel/.style={
    },
    >={Stealth[length=5pt]}
]

\node[spA] (Pcirc) at (0, 0) {$\Psi_\bullet$};
\node[spB] (Pbull) at (3.5, 0) {$\Psi_\blackdiamond$};
\node[spA] (Pbsq) at (7, 0) {$\Psi_\smallblacksquare$};
\node[spB] (Pstar) at (7, -3.5) {$\Psi_\scalestar$};

\draw[thick, ->] (Pcirc) -- node[elabel, above] {$\ket{0,i,0}$} (Pbull);
\draw[thick, ->] (Pbull) -- node[elabel, above] {$\ket{\bot, i, 0}$} (Pbsq);
\draw[thick, ->] (Pbsq) -- (Pstar);
\node[elabel, rotate=-90] at (7.25, -1.75) {$\ket{1, i, 0}$};
\node[elabel, rotate=-90] at (6.75, -1.75) {($f(i) = 0)$};

\draw[thick, ->] (Pbsq) -- ++(3.5, 0) node[midway, above] {$\ket{1,i,0}$} node[midway, below] {$(f(i) = 1)$};

\draw[thick, ->] (-3.5, 0) -- (Pcirc) node[midway, above] {$\ket{\psi_0} = \ket{\circ,0,0}$};

\node[anchor=west] at (-0.8, -5) {$ \Psi_\cA = \Psi_\bullet \cup \Psi_\smallblacksquare$ (solid), \quad $ \Psi_\cB = \Psi_\blackdiamond \cup \Psi_\scalestar$ (dashed)};

\end{tikzpicture}
\caption{Overlap graph for the simpler phase estimation algorithm (no inner transition spaces). Each node represents a set of vectors (defining a subspace), there is an edge between two nodes if and only if they represent non-orthogonal subspaces, and there is a dangling edge from a node if there is a degree of freedom of $H$ for which that node corresponds to the only space that has non-zero overlap with it (the directions on the edges are not meaningful, but rather intuitive). 
Solid circles represent the spaces that make up $\cal A$, and dashed circles $\cal B$. Since the solid circles form an independent set (no two solid-circle nodes are adjacent) they represent non-overlapping spaces, so since each of $\Psi_\bullet$ and $\Psi_\smallblacksquare$ contains pairwise orthogonal vectors, so does $\Psi_{\cal A}$. The same is true for dashed circles and $\Psi_{\cal B}$. 
We should interpret this graph as a kind of ``loop'', where control flow enters from the left, when it gets to the $\Psi_\smallblacksquare$ node, there are two possibilities. If there is some $i$ such that $f(i)=1$, it can exit through the right, though it only takes this option with some small probability. Otherwise, it goes to $\Psi_\star$, which is a dead end, and so it returns to $\Psi_\smallblacksquare$ to try again. 
Intuitively, when there is no $i$ such that $f(i)=1$, there is no way for amplitude to ``flow'' from the initial state $\ket{\psi_0}$ out of the graph, and it is ``trapped'', looping forever. 
In contrast to this ``loop'', the overlap graph of a straight-line program looks like a long (straight) line, with a node for each step.
}
\label{fig:simple-span-graph}
\end{figure}

We now separately analyze the cases when there exists a marked element and when not.

\paragraph{Case 1: Marked Item.}
In the case when there are marked elements, i.e., $\abs{M_f} \geq \mu$, define the following:
\begin{align*}
    \ket{w_+} &:= \ket{\circ,0,0} +  \sum_{i \in M_f}\frac{1}{\abs{M_f}}\sqrt{\frac{N}{\w}} \inparen{\ket{0,i,0} + \ket{\perp, i, 1} + \ket{1,i,1} }. 
\end{align*}
It is clear by inspection that this is orthogonal to all states in $\Psi_\blackdiamond$, $\Psi_{\smallblacksquare}$ and $\Psi_\scalestar$. To see that it is orthogonal to the state in $\Psi_\blackcircle$, we compute the inner product:
\begin{align*}
\braket{\circ,0,0}{\circ,0,0} - \sum_{i\in M_f}\sqrt{\frac{\w}{N}}\frac{1}{|M_f|}\sqrt{\frac{N}{\w}}\braket{0,i,0}{0,i,0} = 1- \frac{|M_f|}{|M_f|} = 0.
\end{align*}
Thus it is orthogonal to both ${\cal A}$ and ${\cal B}$, and therefore a positive witness (see \Cref{def:pos-witness}).
To analyze the complexity of the algorithm, we upper bound its squared norm:
\begin{align*}
    \norm{\ket{w_+}}^2 &= \braket{w_+}{w_+} 
    = 1 + \sum_{i \in M_f} 3\frac{1}{|M_f|^2}\frac{N}{\w}
    = 1+\frac{3N}{|M_f|\w}\leq 1+\frac{3N}{\mu\w}.
\end{align*}
Additionally we have $\braket{w_+}{\psi_0} = 1$. Thus, if we set
$\w = \frac{N}{\mu}$, we get $\frac{\norm{\ket{w_+}}^2}{\braket{w_+}{\psi_0}} \leq c_+$ for the constant $c_+=4$.

\paragraph{Case 0: No Marked Item.}
In the case when there are no marked elements i.e., $M_f=\emptyset$, we define:
\begin{align*}
    \ket{w_\cA} &:= \ket{\circ,0,0} -  \sum_{i \in [N]} \sqrt{\frac{\w}{N}} \ket{0,i,0} +  \sum_{i \in [N]} \sqrt{\frac{\w}{N}} \inparen{\ket{\perp,i,0} - \ket{1,i,0}} \in {\cal A}\\
    \ket{w_\cB} &:=  \sum_{i \in [N]} \sqrt{\frac{\w}{N}} \inparen{\ket{0,i,0} - \ket{\perp, i, 0}} +  \sum_{i \in [N]} \sqrt{\frac{\w}{N}} \ket{1,i,0} \in {\cal B}.
\end{align*}
It is clear from inspection that $\ket{w_\cA} + \ket{w_\cB} = \ket{\psi_0}$, establishing that these vectors do in fact form a negative witness (see \Cref{def:neg-witness}).
To analyze the complexity of the algorithm, we upper bound the squared norm of $\ket{w_{\cal A}}$:
\begin{align*}
    \norm{\ket{w_\cA}}^2 &= \braket{w_\cA}{w_\cA} = 1 + \sum_{i \in [N]} \frac{\w}{N} + \sum_{i \in [N]} \frac{2\w}{N} 
    = 1 + 3\w.
\end{align*}
Again using $\w=\frac{N}{\mu}$,
this is at most $\cC_-:=4N/\mu$.

\paragraph{Conclusion.}
We are now ready to apply \Cref{thm:phase-est-fwk}. To implement the unitary $U_{\cal AB}$, it is sufficient to be able to implement a map that generates normalizations of the vectors in $\Psi_{\cal A}=\Psi_\bullet\cup \Psi_{\smallblacksquare}$, and another map that generates normalizations of the vectors in $\Psi_{\cal B}=\Psi_{\blackdiamond}\cup\Psi_\scalestar$. We leave it an exercise to show that this can be done in complexity $\O(\log N)$ basic gates and queries to $U_f$. Similarly, the initial state can be generated in at most this complexity. Thus, \Cref{thm:phase-est-fwk} implies an algorithm that decides if $|M_f|=0$ (negative case) or $|M_f|\geq \mu$ (positive case) in complexity $\O\left({\sf S}+\sqrt{\cC_-}{\sf A}\right)=\O(\sqrt{N/\mu}\log N)$.
\end{proof}

This result recovers the cost expression one would expect from a standard Grover's search. In the general picture \Cref{sec:general-picture} we study what happens when each element has a different query cost.  We do this by composing subroutines into the loop structure shown in \Cref{fig:simple-span-graph}, whereas \Cref{thm:composition-formal} is proven by composing subroutines into an overlap graph for a straight-line program, that looks, essentially, like a straight line, except with branching behaviour similar to that shown in~\Cref{fig:branching}.

\subsection{Loop composition, with subroutines}
\label{sec:general-picture}

We next turn to the setting where the oracle is implemented by a quantum subroutine -- specifically a variable-time subroutine, as described in  \Cref{sec:VT-model}. Recall that such an algorithm consists of a sequence of unitaries $U_1,\dots,U_{\sf T}$, where we can imagine running the algorithm by applying $U_1$, then measuring $\{\Pi_1,I-\Pi_1\}$ to see if we're done, then if we're not yet done, apply $U_2$ and then measure $\{\Pi_2,I-\Pi_2\}$, etc. until either we measure that the algorithm is done, or we have applied the final unitary $U_{\sf T}$ (then we are definitely done). The number of steps we run the algorithm for on any input $i$ is thus a random variable from $[{\sf T}]'$, which we denote $T_i$. For simplicity, we will assume this subroutine computes $f$ exactly, so with error $\eps=0$.
\begin{theorem}\label{thm:general-picture}
    Consider the problem of \emph{variable-time search}, where the input to the search problem $f: [N] \to \set{0,1}$ is given by a variable-time quantum subroutine that computes $f$ (with no error). Let $T_i$ be the running time of this subroutine on input $i\in [N]$, which is a random variable on $[{\sf T}]$. Let $M_f := \set{i \in [N]: f(i)=1}$ denote the marked set with respect to the function $f$. Let $P_1\subseteq \{f : [N] \to \{0,1\} \ | \  M_{f} \neq \emptyset \}$. Then we have the following.
    \begin{enumerate}
        \item 
    When the expected subroutine costs $\E[{T_i}]$ are known for each $i \in [N]$, there exist bounded-error quantum algorithms that decide whether $M_f=\emptyset$ or $f\in P_1$ in each of the following incomparable complexities:
    \[
        \tO\left(\sqrt{\frac{\sum_{i \in [N]} \E[{T_i}]^2}{\min_{f \in P_1} \abs{M_f}}}\right),\qquad
        \tO\left(\sqrt{\frac{N}{\min_{f\in P_1}\sum_{j\in M_f}\frac{1}{\E[T_j]^2}}} \right).
    \]

    \item
    Even when the costs $\E[T_i]$ are unknown, there exist bounded-error quantum algorithms that decide whether $M_f=\emptyset$ or $f\in P_1$ in each of the following incomparable complexities:
    \[
    \tO\left(\sqrt{\frac{\sum_{i \in [N]} \E[{T_i}^2]}{\min_{f \in P_1} \abs{M_f}}}\right),\quad
        \tO\left(\sqrt{\frac{\sum_{i \in [N]} \E[{T_i}]}{\min_{f \in P_1} \sum_{j \in M_f} \frac{1}{ \E[{T_j}]}}}\right),\quad
        \tO\left(\sqrt{\frac{N}{\min_{f \in P_1} \sum_{j \in M_f} \frac{1}{\E[{T_j}^2]}}}\right).
    \]
    \end{enumerate}

\end{theorem}
The expressions in Item 2 recover the $\ell_2$-, $\ell_1$-, and $\ell_0$-variable-time search expressions discussed in \Cref{sec:VT-bounds} (that is, equations~\eqref{eq:amb-VT-gen},~\eqref{eq:ell1-VT-gen} and~\eqref{eq:ell0-VT-gen} respectively) in the special case where $\pi$ is uniform. 
Note that $\E[T_i^2]-\E[T_i]^2$ is precisely the variance of $T_i$, and when it is large, $\E[T_i]^2$ may be significantly smaller than $\E[T_i^2]$, so when the $\E[T_i]$ are known, we can do slightly better (this was also already shown in~\cite{jeffery2023subroutines}).

\begin{proof}
As in the proof of \Cref{thm:simple-picture}, we start by fixing parameters of a phase estimation algorithm, beginning with defining 
$$H=\mathrm{span}\{\ket{d,i,b,a,z,t}:d\in\{\circ,\bot,0,1,\rightarrow,\leftarrow,\leftrightarrow\},i\in [N]',b\in\{0,1\},a\in\{0,1\},z\in{\cal Z},t\in [{\sf T}]'\},$$
where $(a,z)\in\{0,1\}\times {\cal Z}$ represent the answer register and workspace of the subroutine; $t$ represents the subroutine's program counter, and $d$, $i$, and $b$ serve a similar function as in \Cref{thm:simple-picture}.

Next we define several subspaces, via sets of vectors, that will be combined to define ${\cal A}$ and ${\cal B}$. 
We first define five sets of vectors somewhat similar to the four in the proof of \Cref{thm:simple-picture}, but now we have weights $\{\w_i\}_{i\in [N]}$ that depend on $i$ (to be specified later), and the set $\Psi_\blackdiamond$ has been replaced by two sets, $\Psi_{\rightarrow}$ and $\Psi_{\leftarrow}$:
\begin{equation}
\begin{split}
    \Psi_\bullet &:= \left\{ \ket{\circ} \left(\ket{0,0} - \sum_{i \in [N]} \sqrt{\tfrac{\w_{i}}{N}}\,\ket{i,0}\right) \ket{0,0}\ket{0} \right\} \\
    \Psi_\to &:= \left\{ (\ket{\circ} - \ket{\to})\ket{i,b}\ket{a,0}\ket{0} : a, b \in \set{0,1},\, i \in [N] \right\} \\
    \Psi_\leftarrow &:= \left\{ (\ket{\leftarrow} - \ket{\bot})\ket{i,b}\ket{a,0}\ket{0} : a, b \in \set{0,1},\, i \in [N] \right\} \\
    \Psi_\smallblacksquare &:= \left\{ (\ket{\bot} - \ket{1})\ket{i,b}\ket{a,0}\ket{0} : a, b \in \set{0,1},\, i \in [N] \right\} \\
    \Psi_\scalestar &:= \left\{ \ket{1}\ket{i,0}\ket{a,0}\ket{0} : f(i) = 0,\, i \in [N] \right\}.
\end{split}\label{eq:Psi-non-alg}
\end{equation}
If we now consider the overlap graph of the spaces defined by just these five sets shown in \Cref{fig:disconnected-overlap-graph}, we see that it is disconnected. We now add spaces, based on the subroutines, that ``stitch together'' the two components, by overlapping $\spaan{\Psi_{\leftarrow}}$ and $\spaan{\Psi_\rightarrow}$.

\begin{figure}
\centering
\begin{tikzpicture}[
    spA/.style={circle, draw, thick, minimum size=8mm, inner sep=1pt, font=\footnotesize},
    spB/.style={circle, draw, thick, dashed, minimum size=8mm, inner sep=1pt, font=\footnotesize},
    elabel/.style={font=\tiny, fill=white, inner sep=1pt},
    slabel/.style={font=\scriptsize, text=black!50},
    >={Stealth[length=5pt]}
]

\node[spA] (Pcirc) at (0, 0) {$\Psi_\bullet$};
\node[spB] (Pright) at (2, 0) {$\Psi_\rightarrow$};
\node[spB] (Pleft) at (8,0) {$\Psi_\leftarrow$};
\node[spA] (Pbsq) at (10, 0) {$\Psi_\smallblacksquare$};
\node[spB] (Pstar) at (10, -2) {$\Psi_\scalestar$};

\draw[thick, ->] (-2,0)--(Pcirc);
\draw[thick,->] (Pcirc)--(Pright);
\draw[thick,->] (Pright)--(4,0);
\draw[thick, ->] (6,0)--(Pleft); \draw[thick, ->] (Pleft)--(Pbsq); 
\draw[thick, ->] (Pbsq)--(12,0);
\draw[thick, ->] (Pbsq) -- (Pstar);


\end{tikzpicture}
\caption{The overlap graph of just the spaces defined by the spans of $\Psi_\bullet$, $\Psi_{\rightarrow}$, $\Psi_\leftarrow$, $\Psi_\smallblacksquare$ and $\Psi_\scalestar$ is not connected. Compare with \Cref{fig:simple-span-graph}.}\label{fig:disconnected-overlap-graph}
\end{figure}
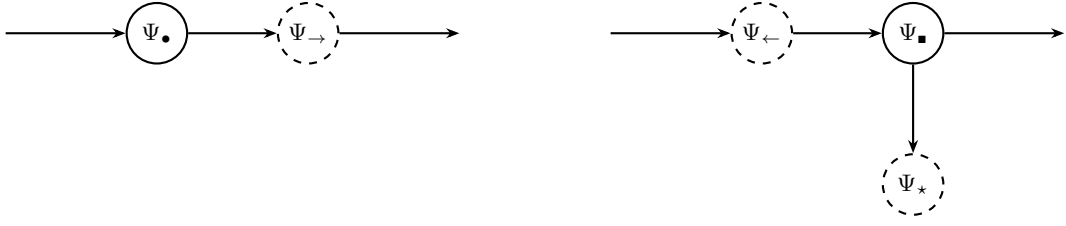

The following \emph{inner transition spaces} are precisely those from \cite{jeffery2023subroutines}\footnote{In fact, they are a special case, where (in the notation of \cite{jeffery2023subroutines}), $A_a\ket{b}=\ket{b\oplus a}$.}.  For $i \in [N]$ and $t \in [{\sf T}]'$, define: 
\begin{align*}
    \Psi^{i,\rightarrow}_{t} &:= \set{ \ket{\rightarrow}\ket{i,b} \inparen{ \sqrt{\a_t} \ket{a,z}\ket{t} - \sqrt{\a_{t+1}} U^i_{t+1} \ket{a,z} \ket{t+1} }: a,b \in \set{0,1}, z \in \cZ_{>t}}
    \\
    \Psi^{i,\leftarrow}_{t} &:= \set{\ket{\leftarrow}\ket{i,b} \inparen{ \sqrt{\a_t} \ket{a,z}\ket{t} - \sqrt{\a_{t+1}} U^i_{t+1} \ket{a,z} \ket{t+1} } 
         : a,b \in \set{0,1}, z \in \cZ_{>t}}
    \\
    \Psi^{i,\leftrightarrow}_{t} &:= \set{ \inparen{\ket{\rightarrow}\ket{i,b} - \ket{\leftarrow} \ket{i,b\oplus a} }\ket{a,z} \ket{t}
        : a,b \in \set{0,1}, z \in \cZ_{t}}.
\end{align*}
Here $U_t^i$ represents the $t$th step of the subroutine on input $i$, and $\{\a_t\}_{t=0}^{\sf T}$ are parameters we can tune, but we always set $\alpha_0=1$.

\Cref{fig:span-graph} shows the overlap graph of the subspaces for a single subroutine~$i$. Each node is a subspace; two nodes are connected when the spaces overlap. Solid circles denote $\cA$, the space spanned by $\Psi_\cA$, and dashed circles denote $\cB$, the space spanned by $\Psi_\cB$.
\begin{figure}[ht]
\centering
\resizebox{\textwidth}{!}{%
\begin{tikzpicture}[
    spA/.style={circle, draw, thick, minimum size=12mm, inner sep=1pt},
    spB/.style={circle, draw, thick, dashed, minimum size=12mm, inner sep=1pt},
    elabel/.style={font=\scriptsize, fill=white, inner sep=1pt},
    slabel/.style={font=\scriptsize, text=black!50},
    >={Stealth[length=5pt]}
]


\node[spA] (Pcirc) at (0, 0) {\Large{$\Psi_\bullet$}};
\node[spB] (Pto) at (4, 0) {\Large{$\Psi_\to$}};
\node[spA] (Pr0) at (8, 0) {\Large{$\Psi_0^{i,\to}$}};
\node[spB] (Pr1) at (12, 0) {\Large{$\Psi_1^{i,\to}$}};
\node[spA] (Pr2) at (16, 0) {\Large{$\Psi_2^{i,\to}$}};
\node at (18, 0) {\Large{$\cdots$}};

\node[spB] (Plr1) at (8, -4) {\Large{$\Psi_1^{i,\leftrightarrow}$}};
\node[spA] (Plr2) at (12, -4) {\Large{$\Psi_2^{i,\leftrightarrow}$}};
\node[spB] (Plr3) at (16, -4) {\Large{$\Psi_3^{i,\leftrightarrow}$}};
\node at (18, -4) {\Large{$\cdots$}};

\node[spA] (Pbull) at (0, -8) {\Large{$\Psi_\smallblacksquare$}};
\node[spB] (Pfrom) at (4, -8) {\Large{$\Psi_\leftarrow$}};
\node[spA] (Pl0) at (8, -8) {\Large{$\Psi_0^{i,\leftarrow}$}};
\node[spB] (Pl1) at (12, -8) {\Large{$\Psi_1^{i,\leftarrow}$}};
\node[spA] (Pl2) at (16, -8) {\Large{$\Psi_2^{i,\leftarrow}$}};
\node at (18, -8) {\Large{$\cdots$}};

\node[spB] (Pstar) at (-3.8, -8) {\Large{$\Psi_\scalestar$}};

\draw[thick, ->] (Pcirc) -- node[elabel, above] {$\ket{\circ}\ket{i,b}\ket{a,0}\ket{0}$} (Pto);
\draw[thick, ->] (Pto) -- node[elabel, above] {$\ket{\to}\ket{i,b}\ket{a,z}\ket{0}$} (Pr0);
\draw[thick, ->] (Pr0) -- node[elabel, above] {$\ket{\to}\ket{i,b}\ket{a,z}\ket{1}$} (Pr1);
\draw[thick, ->] (Pr1) -- node[elabel, above] {$\ket{\to}\ket{i,b}\ket{a,z}\ket{2}$} (Pr2);
\draw[thick, ->] (Pr2) -- (17.5, 0);

\draw[thick, <-] (Pbull) -- node[elabel, below] {$\ket{\bot}\ket{i,b}\ket{a,0}\ket{0}$} (Pfrom);
\draw[thick, <-] (Pfrom) -- node[elabel, below] {$\ket{\leftarrow}\ket{i,b}\ket{a,z}\ket{0}$} (Pl0);
\draw[thick, <-] (Pl0) -- node[elabel, below] {$\ket{\leftarrow}\ket{i,b}\ket{a,z}\ket{1}$} (Pl1);
\draw[thick, <-] (Pl1) -- node[elabel, below] {$\ket{\leftarrow}\ket{i,b}\ket{a,z}\ket{2}$} (Pl2);
\draw[thick, <-] (Pl2) -- (17.5, -8);

\draw[thick, ->] (Pbull) -- node[elabel, above] {$\ket{1}\ket{i,0}\ket{a,0}\ket{0}$} node[elabel, below] {$(f(i) = 0)$} (Pstar);

\draw[thick, ->] (Pbull) -- ++(0, -3);
\node[elabel, rotate=-90] at (0.3, -9.75) {$\ket{1}\ket{i,0}\ket{a,0}\ket{0}$};
\node[elabel, rotate=-90] at (-0.3, -9.75) {$(f(i) = 1)$};

\draw[thick, ->] (Pr0) -- (Plr1);
\node[elabel, rotate=-90] at (8.3, -2) {$\ket{\to}\ket{i,b}\ket{a,z}\ket{1}$};
\draw[thick, ->] (Plr1) -- (Pl0);
\node[elabel, rotate=-90] at (8.3, -6) {$\ket{\leftarrow}\!\ket{i,b\oplus a}\ket{a,z}\ket{1}$};

\draw[thick, ->] (Pr1) -- (Plr2);
\node[elabel, rotate=-90] at (12.3, -2) {$\ket{\to}\ket{i,b}\ket{a,z}\ket{2}$};
\draw[thick, ->] (Plr2) -- (Pl1);
\node[elabel, rotate=-90] at (12.3, -6) {$\ket{\leftarrow}\!\ket{i,b\oplus a}\ket{a,z}\ket{2}$};

\draw[thick, ->] (Pr2) -- (Plr3);
\node[elabel, rotate=-90] at (16.3, -2) {$\ket{\to}\ket{i,b}\ket{a,z}\ket{3}$};
\draw[thick, ->] (Plr3) -- (Pl2);
\node[elabel, rotate=-90] at (16.3, -6) {$\ket{\leftarrow}\!\ket{i,b\oplus a}\ket{a,z}\ket{3}$};

\draw[thick, ->] (0, 2.8) -- (Pcirc);
\node[elabel, rotate=-90] at (0.25, 1.8) { \scriptsize{$\ket{\circ}\ket{0,0}\ket{0,0}\ket{0}$}};

\end{tikzpicture}%
}
\caption{Overlap graph of subspaces for subroutine~$i$. Solid circles represent subspaces of $\cA$ and a dashed circles represent subspaces of $\cB$. Two nodes are adjacent when they represent non-orthogonal spaces. The graph is bipartite --- no two solid nodes or dashed nodes are adjacent. 
}
\label{fig:span-graph}
\end{figure}
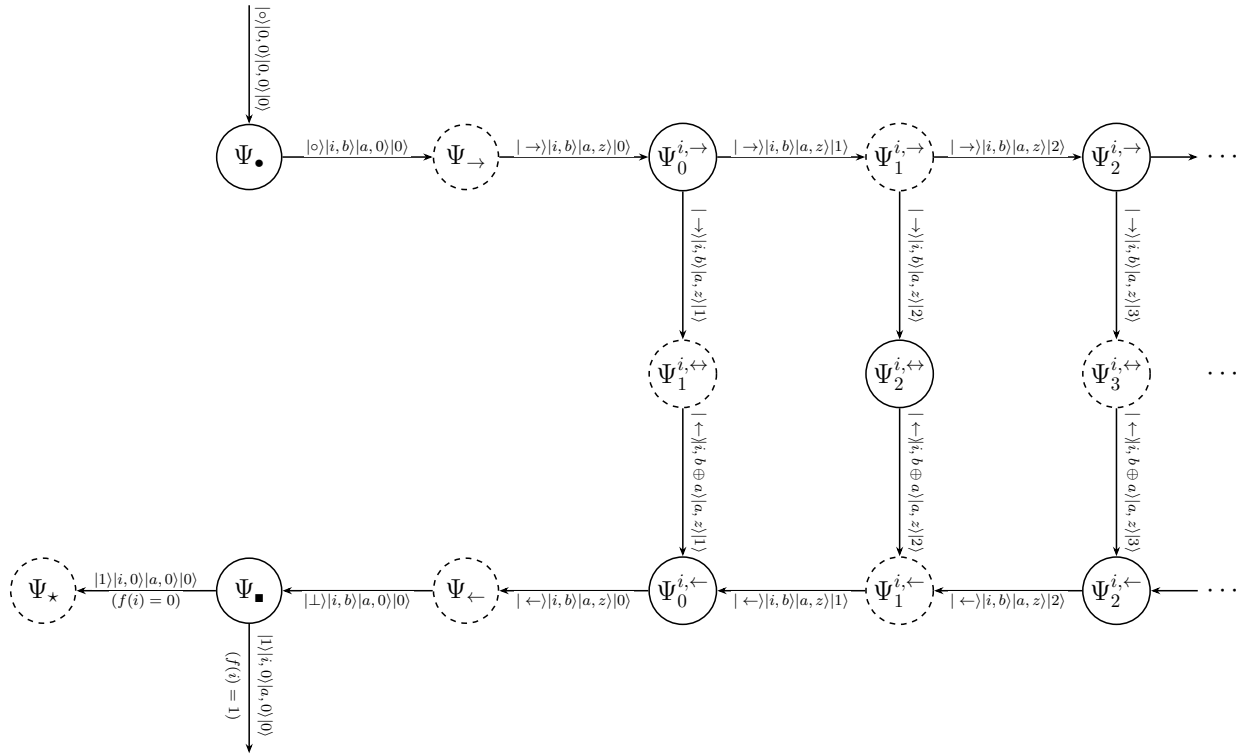
\begin{align*}
    \Psi_\cA = \Psi_\bullet \cup \Psi_{\mathrm{even}} \cup \Psi_\smallblacksquare
    &,\quad 
    \Psi_\cB = \Psi_\to \cup \Psi_{\mathrm{odd}} \cup \Psi_\leftarrow \cup \Psi_\scalestar
    \end{align*}
where
\[
    \Psi_{\mathrm{even}} = \bigcup_{i \in [N]}\bigcup_{t \text{ even}} \left(\Psi_t^{i,\to} \cup \Psi_t^{i,\leftarrow} \cup \Psi_t^{i,\leftrightarrow}\right)
    , \mbox{ and }
    \Psi_{\mathrm{odd}} = \bigcup_{i \in [N]}\bigcup_{t \text{ odd}} \left(\Psi_t^{i,\to} \cup \Psi_t^{i,\leftarrow} \cup \Psi_t^{i,\leftrightarrow}\right) .
\]
Phase estimation is performed on the initial state
\begin{equation}\label{eq:init-state}
    \ket{\psi_0} := \ket{\circ}\ket{0,0}\ket{0,0}\ket{0}.
\end{equation}
We note that $\ket{\psi_0}$ is orthogonal to $\cB$, and that reflection over both spaces can be performed efficiently (see \Cref{lem:reflect-cost} for the proof).

\paragraph{Subroutine states.} To construct witnesses, we define \emph{subroutine states} in $H_A\otimes H_Z$, which encode the computation history of the run of the subroutine on a particular input. For each $i\in [N]$, define:
\begin{align*}
    \ket{w^0(i)}&=\ket{0,0}\\
    \forall t\in [{\sf T}],\; \ket{w^t(i)} &= U_t^i\Pi_{\geq t}\ket{w^{t-1}(i)}.
\end{align*}
For each $i \in [N]$, define:
\begin{align}
    \ket{w_+(i)} &:= (\ket{\to}\ket{i,0} + \ket{\leftarrow}\ket{i, f(i)}) \sum_{t \in [{\sf T}]'} \frac{1}{\sqrt{\a_t}} \ket{w^t(i)}\ket{t} \label{eq:pos-hist}\\
    \ket{w_-(i)} &:= (\ket{\to}\ket{i,0} - \ket{\leftarrow}\ket{i, f(i)}) \sum_{t \in [{\sf T}]'} {(-1)^t}{\sqrt{\a_t}} \ket{w^t(i)}\ket{t} \nonumber\\
    \ket{w_-'(i)} &:= \ket{w_-(i)} - (\ket{\to}\ket{i,0} - \ket{\leftarrow}\ket{i, f(i)})\ket{0,0}\ket{0}.\label{eq:neg-hist-prime}
\end{align}
The states $\ket{w_+(i)}$ and $\ket{w_-(i)}$ are precisely the positive and negative history states defined in \cite[Definition~3.4]{jeffery2023subroutines}.
The state $\ket{w_+(i)}$ is orthogonal to $\spaan{\Psi_\mathrm{even}} \cup \spaan{\Psi_\mathrm{odd}}$, 
$\ket{w_-(i)} \in \spaan{\Psi_\mathrm{even}}$, and $\ket{w_-'(i)} \in \spaan{\Psi_\mathrm{odd}}$. 
More specifically we borrow the result from \cite{jeffery2023subroutines} using the fact that our subroutines are assumed to have no error:

\begin{lemma}[{\cite[Corollary~3.7, Claim~3.8, Claim~3.9]{jeffery2023subroutines}}]
\label{lem:inner-transition-witnesses}
Let $\Pi_{{\mathrm{even}}}$ and $\Pi_{{\mathrm{odd}}}$ denote the orthogonal projectors onto
$\spaan{\Psi_{\mathrm{even}}}$ and $\spaan{\Psi_{\mathrm{odd}}}$ respectively.
Then for all $i \in [N]$:
\begin{enumerate}
    \item \emph{(Norms)} $\displaystyle\norm{\ket{w_+(i)}}^2 = 2\E\insquare{\sum_{t=0}^{T_i} \frac{1}{\a_t}}$ \quad and \quad $\displaystyle\norm{\ket{w_-(i)}}^2 = 2\E\insquare{\sum_{t=0}^{T_i} \a_t}$.

    \item \emph{(Positive witness)} 
$\norm{\Pi_{{\mathrm{even}}} \ket{w_+(i)}}^2 =0$ 
 and $        \norm{\Pi_{{\mathrm{odd}}} \ket{w_+(i)}}^2 =0$. 

    \item \emph{(Negative witness)}    
 $        \norm{(I - \Pi_{{\mathrm{even}}}) \ket{w_-(i)}}^2 =0$        and
        $\norm{(I - \Pi_{{\mathrm{odd}}}) \ket{w'_-(i)}}^2=0$. 
\end{enumerate}
\end{lemma}

These properties allow us to build the positive and negative witnesses for phase estimation of the product of reflections around $\cA$ and $\cB$.

\paragraph{Case 1: Marked item.}
Suppose $f \in P_1$ i.e., $M_f \neq \emptyset$. 
Consider the following positive witness $\ket{w_+}$ defined as
\begin{equation}\label{eq:w-plus}
\begin{split}
    \ket{w_+} &:= \ket{\circ}\ket{0,0}\ket{0,0}\ket{0} \\
    & \quad + \sum_{i \in M_f} \frac{\sqrt{N} \sqrt{\b_i} }{\sqrt{\w_i}}  \Big(\ket{\circ}\ket{i,0}\ket{0,0}\ket{0} + \ket{w_+(i)} 
    + \ket{\perp}\ket{i,1}\ket{0,0} + \ket{1}\ket{i,1}\ket{0,0}\ket{0} \Big) \end{split}
\end{equation}
for some parameters $\set{\b_i}_{i \in M_f}$ that we set later in the analysis, but always ensuring
\begin{align}
    \sum_{i \in M_f} \sqrt{\b_i} &= 1.
    \label{eq:beta-normalize}  
\end{align}
Observe that $\ket{w_+}$ has a unit overlap with the initial state (defined in~\eqref{eq:init-state}): $\braket{w_+}{\psi_0} = 1$. 
To check for its orthogonality with the spaces $\cA$ and $\cB$, we compute the inner product with the state in $\Psi_\blackcircle$:
\begin{align*}
    \braket{\circ, 0, 0, 0, 0, 0}{\circ, 0, 0, 0, 0, 0} - \sum_{i \in M_f} \frac{\sqrt{N \b_i}}{ \sqrt{\w_i}} \sqrt{\frac{\w_i}{N}} \braket{\circ, i, 0, 0, 0, 0}{\circ, i, 0, 0, 0, 0} = 1 -  \sum_{i \in M_f} \sqrt{\b_i} = 0  
\end{align*}
by \eqref{eq:beta-normalize}.
Note that the projection of $\ket{w_+}$ onto states with some specific $i\in [N]$ in the second register, and $t=0$ in the last register is proportional to:
\begin{align*}
\ket{\circ}\ket{i,0}\ket{0,0}\ket{0} + 
    (\ket{\to}\ket{i,0} + \ket{\leftarrow}\ket{i, f(i)})\ket{0,0}\ket{0}
    + \ket{\perp}\ket{i,1}\ket{0,0} + \ket{1}\ket{i,1}\ket{0,0}\ket{0}
\end{align*}
by \eqref{eq:pos-hist}, and using $\ket{w^0(i)}=\ket{0,0}$ and $\alpha_0=1$.
Orthogonality with the states in $\Psi_{\rightarrow}$, $\Psi_\leftarrow$, $\Psi_{\smallblacksquare}$, and $\Psi_{\scalestar}$ (see \eqref{eq:Psi-non-alg}) is then clear from inspection.
For inner transition states in $\Psi_t^{i,\rightarrow}$, $\Psi_{t}^{i,\leftarrow}$ and $\Psi_t^{i,\leftrightarrow}$, orthogonality follows from \Cref{lem:inner-transition-witnesses}.
\begin{align*}
    \norm{\Pi_{{\mathrm{even}}} \ket{w_+}}^2 =  \sum_{i \in M_f} \frac{\sqrt{N\b_i}}{\sqrt{\w_i}} \norm{\Pi_{\mathrm{even}} \ket{w_+(i)}}^2 = 0 
\end{align*}
    and similarly $\norm{\Pi_{{\mathrm{odd}}} \ket{w_+}}^2 = 0$.
Consequently, $\ket{w_+}$ is orthogonal to both spaces $\cA$ and $\cB$.
Using \Cref{lem:inner-transition-witnesses} the size of our positive witness is 
\begin{align}
    \norm{\ket{w_+}}^2 
    &= 1 + N  \sum_{i \in M_f} \frac{\b_i}{\w_i} \inparen{ 1 +  \norm{\ket{w_+(i)}}^2  + 2}
    = 1 + N  \sum_{i \in M_f} \frac{\b_i}{\w_i} \inparen{ 3 +  2\E\insquare{\sum_{t=0}^{{T_i}} \frac{1}{\a_t}} }. 
    \label{eq:pos-witness-size} 
\end{align}

\paragraph{Case 0: No marked item.}
We present a negative witness in the case when $f(i) = 0$ for all $i$. Define:
\begin{align*}
    \ket{w_\cA} &= \ket{\circ}\ket{0,0}\ket{0,0}\ket{0} +  \sum_{i \in [N]} \sqrt{\frac{\w_i}{N}} \inparen{ - \ket{\circ}\ket{i,0}\ket{0,0}\ket{0} 
    + \ket{w_-(i)}
    + \inparen{\ket{\perp} - \ket{1}} \ket{i,0}\ket{0,0}\ket{0}}.  
\end{align*}
Here the first two terms together are in $\spaan{\Psi_\blackcircle}$ and the fourth term is in $\spaan{\Psi_{\smallblacksquare}}$ by inspection (see \eqref{eq:Psi-non-alg}), and the third term is in $\spaan{\Psi_{ \text{even} }}$ according to \Cref{lem:inner-transition-witnesses}.
Similarly if we define $\ket{w_\cB}$ and separate into terms as follows
\begin{align*}
    \ket{w_{\cB}}   : = \ket{\psi_0} - \ket{w_\cA}
    &= \sum_{i \in [N]} \sqrt{\frac{\w_{i}}{N}} \inparen{\ket{\circ} - \ket{\rightarrow}}\ket{i,0}\ket{0,0}\ket{0}
    \\
    & \quad + \sum_{i \in [N]} \sqrt{\frac{\w_{i}}{N}} \inparen{ \inparen{\ket{\rightarrow} - \ket{\leftarrow}}\ket{i,0}\ket{0,0}\ket{0}
    - \ket{w_-(i)}}
    \\
    & \quad + \sum_{i \in [N]} \sqrt{\frac{\w_{i}}{N}} \inparen{\ket{\leftarrow} - \ket{\perp}}\ket{i,0}\ket{0,0}\ket{0}
     + \sum_{i \in [N]} \sqrt{\frac{\w_{i}}{N}} \ket{1}\ket{i,0}\ket{0,0}\ket{0},
\end{align*}
the first, third and fourth terms are in the spans of $\Psi_{\rightarrow}$, $\Psi_{\leftarrow}$ and $\Psi_{\scalestar}$ respectively and the second term, which  is inside the parentheses, is $ -\sum_i\sqrt{\w_i/N}\ket{w'_-(i)}$ (see \eqref{eq:neg-hist-prime}), which, is in $\spaan{\Psi_{ \mathrm{odd} }}$ from \Cref{lem:inner-transition-witnesses}. 
Therefore we have $\ket{w_\cA} \in \cA$ and $\ket{w_\cB} \in \cB$.
The size of the negative witness, again from \Cref{lem:inner-transition-witnesses}, is:
\begin{align}
    \norm{\ket{w_\cA}}^2 
    &= 1 + \frac{1}{N}  \sum_{i \in [N]} \w_i \inparen{1 + \norm{\ket{w_-(i)}}^2 + 2}
= 1 + \frac{1}{N}  \sum_{i \in [N]} \w_i \inparen{3 + 2 \E \insquare{ \sum_{t=0}^{T_i}  \a_t} }.
    \label{eq:neg-witness-size}
\end{align}

\paragraph{Costs of steps of phase estimation.}
In order to be able to use \Cref{thm:phase-est-fwk} we need to analyze the cost ${\sf S}$, of preparing our initial state $\ket{\psi_0}$, which is easily seen to be at most $O(\log N{\sf T})$; and the cost $\sf A$, of performing a single application of $U_{\cal AB}=(2\Pi_{\cal A}-I)(2\Pi_{\cal B}-I)$. The latter is also $O(\log(N{\sf T}))$, as we show in the following lemma.

\begin{lemma}
    Each of $2\Pi_{\cal A}-I$ and $2\Pi_{\cal B}-I$ can be implemented in $O(\log (N {\sf T}))$ cost.
    \label{lem:reflect-cost}
\end{lemma}

\begin{proof}
To derive the cost of performing a reflection around the spaces $\cA$ and $\cB$, it is enough to just look at the costs of reflections in the individual orthogonal subspaces.
If $\Pi_\blackcircle$, $\Pi_{\smallblacksquare}$ and $\Pi_{\mathrm{even}}$ are the projectors onto the spans of $\Psi_\blackcircle$, $\Psi_\smallblacksquare$ and $\Psi_{\mathrm{even}}$ respectively, the projector onto $\cA$, $\Pi_\cA$ can be written as
\begin{align*}
    \Pi_\cA &= \Pi_\blackcircle + \Pi_{\smallblacksquare} + \Pi_{\mathrm{even}}
    \intertext{which then allows us to write the reflection around $\cA$ as}
    2\Pi_{\cA} - I &= \inparen{2\Pi_\blackcircle - I} \inparen{2\Pi_\smallblacksquare - I} \inparen{2 \Pi_{\mathrm{even}} - I}
\end{align*}
because the product of projectors onto orthogonal spaces is zero. To perform a reflection around $\cA$ therefore  it is enough to perform individual reflections in the spans of $\Psi_\blackcircle$, $\Psi_\smallblacksquare$ and $\Psi_{\mathrm{even}}$ sequentially. 
Reflection around $\spaan{\Psi_\blackcircle}$ can be performed by implementing the unitary $G_\blackcircle : \ket{\circ}\ket{0,b
}\ket{a,z}\ket{t} \to \frac{1}{\sqrt{1+\w_i/N}} \ket{\circ}\inparen{\ket{0,b} -  \sum_{i\in [N]} \sqrt{\frac{\w_i}{N}} \ket{i,b}  }\ket{a,z}\ket{t}$ controlled on $\ket{\circ}$ in the first register, 
and performing the reflection around the zero state in the second register. More precisely we perform
\begin{align*}
    2 \Pi_\blackcircle - I &= G_\blackcircle (I \otimes \inparen{ 2\ketbra{0}{0} - I } \otimes I )G_\blackcircle^\dagger.
\end{align*}
This involves $O(\log N)$ basic operations. Similarly the reflection around  $\spaan{\Psi_\smallblacksquare}$ can be performed by implementing unitary $G_{\smallblacksquare}: \ket{0,i,b,a,z,t} \to \frac{1}{\sqrt{2}}(\ket{\perp}-\ket{1}) \ket{i,b,a,z,t}$ 
 and a reflection about the zero state in the first register. More precisely we do
\begin{align*}
    2 \Pi_\smallblacksquare - I &= G_\smallblacksquare ((2 \ketbra{0}{0}-I) \otimes I) G_\smallblacksquare^\dagger.
\end{align*}
This costs $O(1)$ basic operations. The reflection around $\spaan{\Psi_{\mathrm{even}}}$ can be implemented in $O(\log {\sf T})$ cost using \cite[Lemma 3.3]{jeffery2023subroutines} and our ability to perform $U_t$ in unit cost for superposition over $t \in [{\sf T}]'$. This results in a combined cost of $O(\log (N {\sf T}))$ for the reflection around $\cA$. A similar analysis holds for $\cB$, where an additional control on sixth register being $t=0$ for the reflection around $\spaan{\Psi_\rightarrow}$ and $\spaan{\Psi_\leftarrow}$ incurs an additive $O(\log {\sf T})$ cost.
\end{proof}

\paragraph{Cost analysis for different settings of parameters.}
Here we analyze the witness sizes for a few different settings of the parameters $\w_i$ and $\a_t$, in order to give the five upper bounds claimed in the theorem statement. We begin with the case where the costs are known.

\paragraph{(i) Known costs.} When the costs $\E[T_i]$ are known, in the sense that they can be efficiently computed from each $i$, then the algorithm, and in particular the weights $\w_i$, can depend on them. In this regime, we do not need the weights $\a_t$, so we set $\a_t = 1$ for all $t \in [{\sf T}]$, which makes $\E \insquare{\sum_{t=0}^{T_i} \frac{1}{\a_t}  } = \E \insquare{  \sum_{t=0}^{T_i} \a_t  } = \E [T_i]+1$. 
\paragraph{(i a).} First, set $\w_i = \frac{N}{ \mu} \E [T_i]$ for all $i \in [N]$, where $\mu = \min_{f\in P_1} \abs{M_f}$. We choose the positive witness $\ket{w_+}$ by setting $\b_i = \frac{1}{\abs{M_f}^2}$ for all $i \in M_f$ in \eqref{eq:w-plus}, which satisfies \eqref{eq:beta-normalize}, since $\sum_{i\in M_f}\sqrt{\b_i}=1$.
Then, referring to \eqref{eq:pos-witness-size}, the witness size can be upper bounded by:
\begin{align*}
    \norm{\ket{w_+}}^2 &= 1 + N  \sum_{i\in M_f} \frac{\mu}{N \abs{M_f}^2 \E[T_i]} \inparen{3 + 2 \E [T_i]+2}
    \\
    &\leq 1 + 7\frac{\mu}{\abs{M_f}} \leq 8, 
\end{align*}
resulting in $\frac{\braket{w_+}{\psi_0}}{\norm{\ket{w_+}}^2} \geq \frac{1}{c_+}$ for $c_+ = 8$. The size of the negative witness from \eqref{eq:neg-witness-size} can be upper bounded by $\cC_-$ which is
\begin{align*}
    \norm{\ket{w_\cA}}^2 &= 1 +  \frac{1}{N}\sum_{i \in [N]} \frac{N\E[T_i]}{\mu} \inparen{3 + 2\E[T_i]+2}
    = \O \inparen{\frac{ \sum_{i \in [N]} \E[T_i]^2 }{\mu}}, 
\end{align*}
so there is some value $\cC_-=O\left(\frac{1}{\mu}\sum_{i \in [N]} \E[T_i]^2\right)$ that always upper bounds $\norm{\ket{w_{\cal A}}}^2$. 
Thus, by \Cref{thm:phase-est-fwk} and \Cref{lem:reflect-cost}, there is a bounded-error quantum algorithm that distinguishes the all 0 case from $f\in P_1$ with complexity:
\begin{align}
    \O \inparen{\sqrt{\frac{ \sum_{i\in[N]} \E[T_i]^2  }{\min_{f \in P_1} \abs{M_f}}} \log \inparen{N {\sf T}}}.
    \label{eq:cost-exp-1}
\end{align}

\paragraph{(i b).} Alternatively, set $\w_i = \frac{N}{k\E [T_i]}$ for all $i \in [N]$, where $k=\sum_{j\in M_f}\frac{1}{\E[T_j]^2}$. We choose the positive witness $\ket{w_+}$ by setting $\b_i = \frac{1}{\E[T_i]^4k^2}$ for all $i \in M_f$ in \eqref{eq:w-plus}, which satisfies \eqref{eq:beta-normalize}, since $\sum_{i\in M_f}\sqrt{1/\E[T_i]^4}=k$.
Then, referring to \eqref{eq:pos-witness-size}, the witness size can be upper bounded by:
\begin{align*}
    \norm{\ket{w_+}}^2 &= 1 + N  \sum_{i\in M_f} \frac{k\E[T_i]}{N\E[T_i]^4k^2} \inparen{3 + 2 \E [T_i]+2 }
    \\
    &\leq 1 + 7\frac{\sum_{i\in M_f}\frac{\E[T_i]^2}{\E[T_i]^4}}{k} \leq 8, 
\end{align*}
resulting in $\frac{\braket{w_+}{\psi_0}}{\norm{\ket{w_+}}^2} \geq \frac{1}{c_+}$ for $c_+ = 8$. The size of the negative witness from \eqref{eq:neg-witness-size} can be upper bounded by $\cC_-$ which is
\begin{align*}
    \norm{\ket{w_\cA}}^2 &= 1 +\frac{1}{N}  \sum_{i \in [N]} \frac{N}{k\E[T_i]} \inparen{3 + 2\E[T_i]+2}
    = \O \inparen{\frac{ N}{k}}.
\end{align*}
Thus, by \Cref{thm:phase-est-fwk} and \Cref{lem:reflect-cost}, there is a bounded-error quantum algorithm that distinguishes the all 0 case from $f\in P_1$ with complexity:
\begin{align}
    \O \inparen{\sqrt{\frac{N}{k}}\log(N{\sf T})} = \O\left(\sqrt{\frac{N}{ \sum_{i\in M_f} \frac{1}{\E[T_i]^2}}  } \log \inparen{N {\sf T}}\right).
    \label{eq:cost-exp-1b}
\end{align}

\paragraph{(ii) Unknown costs.} In the case when the costs $\E[T_i]$ are unknown, the weight $\w_i$, which defines the spaces around which we reflect, cannot depend on $\E [T_i]$ -- it will actually be independent of $i$ in this setting. We can compensate for this by using different settings of the weights $\a_t$, as we will see shortly.
The parameter $\b_i$ defining the positive witness may also depend on the value $\E[T_i]$, as it only appears in the analysis and the algorithm itself doesn't depend on $\b_i$.

\paragraph{(ii a).}
First, we choose the setting $\a_t = t+1$, which results in $\E \insquare{ \sum_{t =0}^{T_i} \frac{1}{\a_t} } \leq 2\E [\log T_i]$ and $\E \insquare{ \sum_{t=0}^{T_i} \a_t} \leq \E[T_i^2]$.
We set $\w_i = \frac{N\log {\sf T}}{\mu}$ where $\mu=\min_{f' \in P_1} \abs{M_{f'}}$ and choose uniform $\b_i = \frac{1}{\abs{M_f}^2}$ (so $\sqrt{\b_i} = 1/\abs{M_f}$ and $\sum_{i\in M_f}\sqrt{\b_i} = 1$ as required by \eqref{eq:beta-normalize}).
Using \eqref{eq:pos-witness-size}:
\begin{align*}
    \norm{\ket{w_+}}^2
    &\leq 1 + N \sum_{i \in M_f} \frac{\mu \b_i}{N\log {\sf T}}\inparen{3 + 4\E[\log T_i]}
    \;\leq\; 1 + \frac{\mu}{\log{\sf T}}\sum_{i\in M_f}\frac{1}{|M_f|^2}7\E[\log T_i]\\
    &\leq 1+7\frac{\mu|M_f|}{|M_f|^2}\leq 8,
\end{align*}
so $\frac{\braket{w_+}{\psi_0}}{\norm{\ket{w_+}}^2} \geq \frac{1}{c_+}$ with $c_+ = 8$. The negative-witness size from~\eqref{eq:neg-witness-size} is
\begin{align*}
    \norm{\ket{w_\cA}}^2
    &\leq 1 + \frac{1}{N}\cdot\frac{N\log {\sf T}}{\mu}\sum_{i\in[N]}\inparen{3 + 2\E[T_i^2]}
    \;=\; \O \inparen{\frac{\sum_{i\in[N]} \E[T_i^2]}{\mu}\log {\sf T}}.
\end{align*}
Thus, by \Cref{thm:phase-est-fwk} and \Cref{lem:reflect-cost}, there is a bounded-error quantum algorithm that distinguishes the all 0 case from $f\in P_1$ with complexity:
\begin{align}
    \O \inparen{\sqrt{\frac{ \sum_{i\in[N]} \E[T_i^2] }{\mu}\log{\sf T}} \log \inparen{N {\sf T}}}
    = \tO \inparen{\sqrt{\frac{\sum_{i\in[N]} \E[T_i^2]}{\min_{f'\in P_1} \abs{M_{f'}}}}}.
    \label{eq:cost-exp-2a}
\end{align}
\paragraph{(ii b).}  
We can get a different upper bound by choosing $\a_t = 1$ for all $t$, which results in $\E \insquare{ \sum_{t=0}^{T_i} \frac{1}{\a_t} } = \E \insquare{ \sum_{t=0}^{T_i} \a_t} = \E [T_i]+1$, just as in setting \textbf{(i)}.
However, unlike in setting \textbf{(i)}, the weights $\w_i$ cannot depend on $\E[T_i]$, so instead we set $\w_i = \frac{N}{k}$ for 
$k = \min_{f' \in P_1} \sum_{j \in M_{f'}}\frac{1}{\E [T_j]}$, and choose $\b_i = \inparen{\frac{\frac{1}{\E[T_i]}}{ \sum_{j \in M_f} \frac{1}{\E[T_j]}}  }^2$, which ensures that \eqref{eq:beta-normalize} is satisfied, since $\sum_{i\in M_f}\sqrt{\beta_i}=1$.
Write $S_f := \sum_{j \in M_f} 1/\E[T_j]$, so $k \le S_f$ for every $f\in P_1$. Using \eqref{eq:pos-witness-size}:
\begin{align*}
    \norm{\ket{w_+}}^2
    &\leq 1 + N \sum_{i \in M_f} \frac{k\,\b_i}{N}\inparen{3 + 2\E[T_i]+2}
    \;=\; 1 + \frac{k}{S_f^2}\sum_{i\in M_f}\frac{7\E[T_i]}{\E[T_i]^2}
    =1 + \frac{7k}{S_f^2}\sum_{i\in M_f}\frac{1}{\E[T_i]}
    \\
    &\leq 1 + \frac{7}{S_f}\sum_{i\in M_f}\frac{1}{\E[T_i]}
    \qquad \text{(using $k \leq S_f$)}
    \\
    &=1+7=8,
\end{align*}
so $\frac{\braket{w_+}{\psi_0}}{\norm{\ket{w_+}}^2} \geq \frac{1}{c_+}$ with $c_+ = 8$. The negative witness size from~\eqref{eq:neg-witness-size} is
\begin{align*}
    \norm{\ket{w_\cA}}^2
    &= 1 + \frac{1}{N}\cdot\frac{N}{k}\sum_{i\in[N]}\inparen{3 + 2\E[T_i]+2}
    \;=\; \O \inparen{\frac{\sum_{i\in[N]} \E[T_i]}{k}}.
\end{align*}
Thus, by \Cref{thm:phase-est-fwk} and \Cref{lem:reflect-cost}, there is a bounded-error quantum algorithm that distinguishes the all 0 case from $f\in P_1$ with complexity:
\begin{align}
    \O \inparen{\sqrt{\frac{ \sum_{i\in[N]} \E[T_i]  }{\min_{f' \in P_1} \sum_{j \in M_{f'}} 1/\E[T_j]}} \log \inparen{N {\sf T}}}
    =
    \tO \inparen{\sqrt{\frac{ \sum_{i\in [N]} \E[T_i] }{\min_{f'\in P_1}  \sum_{j\in M_{f'}} \frac{1}{\E[T_j]} }}}.
    \label{eq:cost-exp-2b}
\end{align}
\paragraph{(ii c).}  
Finally, we can get a different complexity by setting $\a_t = \frac{1}{t+1}$ , which results in $\E \insquare{ \sum_{t \in [T_i]'} \frac{1}{\a_t} } \leq \E [T_i^2]$ and $\E \insquare{ \sum_{t\in[T_i]'} \a_t} \leq 2\E[\log T_i]$.
We set $\w_i = \frac{N}{ k }$ for
$k = \min_{f' \in P_1} \sum_{j \in M_{f'}}\frac{1}{\E [T_j^2]}$ and use $\b_i = \inparen{\frac{\frac{1}{\E[T_i^2]}}{ \sum_{j \in M_f} \frac{1}{\E[T_j^2]}}  }^2$, which satisfies \eqref{eq:beta-normalize} since $\sum_{i\in M_f}\sqrt{\b_i}=1$.
Write $S_f := \sum_{j \in M_f} 1/\E[T_j^2]$, so $k \le S_f$ for every $f\in P_1$. Using \eqref{eq:pos-witness-size}:
\begin{align*}
    \norm{\ket{w_+}}^2
    &\leq 1 + N \sum_{i \in M_f} \frac{k\,\b_i}{N}\inparen{3 + 2\E[T_i^2]}
    \;\leq\; 1 + \frac{k}{S_f^2}\sum_{i\in M_f}\frac{5\E[T_i^2]}{\E[T_i^2]^2}
    \\
    &= 1 + \frac{5k}{S_f^2}\sum_{i\in M_f}\frac{1}{\E[T_i^2]}
    \qquad\qquad\quad\leq 1 + \frac{5}{S_f}\sum_{i\in M_f}\frac{1}{\E[T_i^2]} 
    \qquad\text{(using $k \leq S_f$)}
    \\
    &= 1 + 5 \;=\; 6,
\end{align*}
so $\frac{\braket{w_+}{\psi_0}}{\norm{\ket{w_+}}^2} \geq \frac{1}{c_+}$ with $c_+ = 6$. The negative-witness size from~\eqref{eq:neg-witness-size} is
\begin{align*}
    \norm{\ket{w_\cA}}^2
    &= 1 + \frac{1}{N}\cdot\frac{N}{k}\sum_{i\in[N]}\inparen{3 + 4\E[\log T_i]}
    \;=\; \O \inparen{\frac{\sum_{i\in[N]} \E[\log T_i]}{k}}
    \;=\; \O \inparen{\frac{N \log {\sf T}}{k}}.
\end{align*}
Thus, by \Cref{thm:phase-est-fwk} and \Cref{lem:reflect-cost}, there is a bounded-error quantum algorithm that distinguishes the all 0 case from $f\in P_1$ with complexity:
\begin{equation}
    \O \inparen{\sqrt{\frac{ N \log {\sf T}  }{k}} \log \inparen{N {\sf T}}}
    = \tO \inparen{ \sqrt{\frac{N}{\min_{f\in P_1}  \sum_{j \in M_f} \frac{1}{\E[T_j^2]} }}}.
    \qedhere
    \label{eq:cost-exp-2c}
\end{equation}
\end{proof}

\bibliographystyle{alpha}
\bibliography{refs}

\end{document}